\pdfoutput=1

\documentclass[a4paper,USenglish]{lipics-v2021}

\hideLIPIcs
\nolinenumbers

\usepackage{stmaryrd}
\usepackage{amsmath}
\usepackage{amsthm}

\usepackage{xspace,amssymb,amsmath}
\usepackage{algpseudocode}
\usepackage{xcolor}
\usepackage{tikz}

\newtheorem{property}{Property}
\newtheorem{notation}{Notation}

\usepackage{hyperref} 
\usepackage{url}
\usepackage{doi} 

\newcommand{\claimend}{$\hfill\lhd$\medskip{}}
\newcommand{\UU}{\mathcal{U}}

\hypersetup{
  colorlinks=true,
  linkcolor=black,
  citecolor=black,
  filecolor=black,
  urlcolor=[rgb]{0,0.1,0.5},
  pdftitle={Robust maximal matching},
  pdfauthor={}
}

\title{When Should You Wait Before Updating?}
\subtitle{Toward a Robustness Refinement}
\date{}

\author{Swan Dubois}
{Sorbonne Université, CNRS, LIP6, DELYS, France}%
{swan.dubois@lip6.fr}%
{https://orcid.org/0000-0003-2320-6178}%
{}

\author{Laurent Feuilloley}
{Univ Lyon, CNRS, INSA Lyon, UCBL, LIRIS, UMR5205, Villeurbanne, France }%
{laurent.feuilloley@cnrs.fr}%
{https://orcid.org/0000-0002-3994-0898}%
{}

\author{Franck Petit}
{Sorbonne Université, CNRS, LIP6, DELYS, France}%
{franck.petit@lip6.fr}%
{https://orcid.org/0000-0002-0948-7842}%
{}

\author{Mika\"el Rabie}
{Université Paris Cité, CNRS, IRIF, Paris, France}%
{mikael.rabie@gmail.com}%
{}%
{}

\funding{ANR SKYDATA (ANR-22-CE25-0008-02) ANR GrR (ANR-18-CE40-0032)}



\authorrunning{S. Dubois, L. Feuilloley, F. Petit and M. Rabie}

\begin{document}
\maketitle

\begin{abstract}
Consider a dynamic network and a given distributed problem. 
At any point in time, there might exist several solutions that are equally good with respect to the problem specification, but that are different from an algorithmic perspective, because some could be  easier to update than others when the network changes. 
In other words, one would prefer to have a solution that is more robust to topological changes in the network; and in this direction the best scenario would be that the solution remains correct despite the dynamic of the network.

In~\cite{CasteigtsDPR20}, the authors introduced a very strong robustness criterion: they required that for any removal of edges that maintain the network connected, the solution remains valid. 
They focus on the  maximal independent set problem, and their approach consists in characterizing the graphs in which there exists a robust solution (the existential problem), or even stronger, where any solution is robust (the universal problem). 
As the robustness criteria is very demanding, few graphs have a robust solution, and even fewer are such that all of their solutions are robust. 
In this paper, we ask the following question: \textit{Can we have robustness for a larger class of networks, if we bound the number $k$ of edge removals allowed}? 

To answer this question, we consider three classic problems: maximal independent set, minimal dominating set and maximal matching. 
For the universal problem, the answers for the three cases are surprisingly different. 
For minimal dominating set, the class does not depend on the number of edges removed. For maximal matching, removing only one edge defines a robust class related to perfect matchings, but for all other bounds $k$, the class is the same as for an arbitrary number of edge removals.
Finally, for maximal independent set, there is a strict hierarchy of classes: the class for the bound $k$ is strictly larger than the class for bound $k+1$.

For the robustness notion of \cite{CasteigtsDPR20}, no characterization of the class for the existential problem is
known, only a polynomial-time recognition algorithm. We show that the situation is even worse for bounded $k$:
even for $k=1$, it is NP-hard to decide whether a graph has a robust maximal independent set. 
\end{abstract}



\section{Introduction}


In the field of computer networks, the phrase ``\textit{dynamic networks}'' refers to many different realities, ranging
from static wired networks in which links can be unstable, up to wireless ad hoc networks in which entities directly communicate with each other by radio.  
In the latter case, entities may join, leave, or even move inside the network at any time in completely unpredictable ways. 
A common feature of all these networks is that communication links keep changing over time. 
Because of this aspect, algorithmic engineering is far more difficult than in fixed static networks.  
Indeed, solutions must be able to
adapt to incessant topological changes. This becomes particularly challenging when it comes to maintaining a single
leader~\cite{CF13r} or a (supposed to be) ``static'' covering data structure, for instance, a spanning tree, a node
coloring, a Maximal Independant Set (MIS), a Minimal Dominating Set (MDS), or a Maximal Matching (MM).
Most of the time, to overcome such topological changes, algorithms compute and recompute their solution to try to be as close as possible to a correct solution in all circumstances.  

Of course, when the network dynamics is high, meaning that topological changes are extremely frequent, it sometimes becomes impossible to obtain an acceptable solution. 
In practice, the correctness requirements of the algorithm are most often relaxed in order to approach the desired behavior, while amortizing the recomputation cost.  
Actually, this sometimes leads to reconsider the very nature of the
problems, for example: looking for a ``moving leader'', a leader or a spanning tree per connected component, a temporal
dominated set, an evolving MIS, a best-effort broadcast, \textit{etc.}--- we refer
to~\cite{CF13r,CFQS12} for more examples.

In this paper, we address the problem of network dynamics under an approach similar to the one introduced
in~\cite{BDKP15,CasteigtsDPR20}: \textit{To what extent of network dynamics can a computation be performed without relaxing its specification?}  
Before going any further into our motivation, let us review related work on which our study relies.  

Numerous models for dynamic networks have been proposed during the last decades--refer to~\cite{CF13r} for a
comprehensive list of models-- some of them aiming at unifying previous modeling approaches, mainly~\cite{CFQS12,XFJ03}.
As is often the case, in this work, the network is modeled as a graph, where the set of vertices (also called nodes) is fixed, while the communication
links are represented by a set of edges appearing and disappearing unexpectedly over the time. 
Without extra assumptions, this modeling includes all possibilities that can occur over the time, for example, the
network topology may include no edges at some instant, or it may also happen that some edge present at some time disappears definitively after that. 
According to different assumptions on the appearance and disappearance (frequency, synchrony, duration, etc.), the dynamics of temporal networks can be classified in many classes~\cite{CFQS12}. 

One of these classes, Class $\mathcal{TC^R}$, is particularly
important. 
In this class, a temporal path between any two vertices appears infinitely often. 
This class is arguably the most natural and versatile generalization of the notion of connectivity from static networks to dynamic networks: every vertex is able to send (not necessarily directly) a message to any other vertex at any time.  

For a dynamic network of the class $\mathcal{TC^R}$ on a vertex set $V$, one can partition $V\times V$ into three sets: the edges that are present infinitely often over the time --called \emph{recurrent} edges--, the edges that are present only a finite number of times --called \emph{eventually absent} edges--, and the edges that are never present. 
The union of the first two sets defines a graph called the {\em footprint} of the network~\cite{CFQS12},
while its restriction to the edges that are infinitely often present is called the {\em eventual footprint}~\cite{BDKP16}.
In~\cite{BDKP16}, the authors prove that Class~$\mathcal{TC^R}$ is actually the set of dynamic networks whose {\em eventual footprint} is connected. 

In conclusion, from a distributed computing point of view, it is more than reasonable to consider only dynamic networks such that some of their edges are recurrent and their union does form a {\em connected} spanning subgraph of their footprint.

Unfortunately, it is impossible for a node to distinguish between a recurrent and an eventually absent edge. Therefore, the best the nodes can do is to compute a solution relative to the footprint, hoping that this solution 
still makes sense in the eventual footprint, whatever it is. 
In \cite{CasteigtsDPR20}, the authors introduce the concept of \textit{robustness} to capture this intuition, defined as follows: 
\begin{definition}[Robustness]
A property $P$
is robust over a graph $G$ if and only if $P$ is satisfied in every connected spanning subgraph of $G$ (including $G$
itself).  
\end{definition}

Another way to phrase this definition is to say that \emph{a property $P$ is robust if it is still satisfied when we remove any number of edges, as long as the graph stays connected}.

In~\cite{CasteigtsDPR20}, the authors  focus on the problem of maximal independent set (MIS). That is, they study the cases where a set of vertices can keep being an MIS even if we remove edges. They structure their results around two questions:
\medskip

\noindent\textbf{Universal question:} For which networks are \emph{all the solutions} robust against any edge removals that do not disconnect the graph?

\medskip

\noindent\textbf{Existential question:} For which networks does \emph{there exist a solution} that is robust against any edge removals that do not disconnect the graph?

\medskip

The authors in~\cite{CasteigtsDPR20} establish a characterization of the networks that answer the first questions for
the MIS problem. Still for the same problem, they provide a polynomial-time algorithm to decide whether a network answers the second question.

Note that the study of robustness was also very recently addressed for the case of metric properties in~\cite{CasteigtsCHL22}.
In that paper, the authors show that deciding whether the distance between two given vertices is robust can be done in
linear time.  However, they also show that deciding whether the diameter is robust or not is coNP-complete. 

\subsection{Our approach}

Our goal is to go beyond \cite{CasteigtsDPR20}, and to get both a more fine-grained and a broader understanding of the notion of robustness. 

Let us start with the fine-grain dimension. 
In~\cite{CasteigtsDPR20}, a solution had to be robust against any number of edge removals as long as the graph remains connected. 
In this paper, we want to understand what are the structures that are robust against $k$ edge removals while keeping the connectivity constraint, for any specific $k$, adding granularity to the notion.
We call this concept $k$-robustness (see formal definition below) and we focus on the universal and the existential question of \cite{CasteigtsDPR20} for this fined-grain version of the robustness.


Now for the broader dimension, let us discuss the problems studied. In~\cite{CasteigtsDPR20}, the problem studied is MIS, which is a good choice in the sense that it leads to
an interesting landscape. 
Indeed, robustness being a very demanding property, one has to find problems to which it can
apply without leading to trivial answers. 
In this direction, one wants to look at local problems, because a modification will only have consequences in some
neighborhood and not on the whole graph, which leaves the hope that it actually does not affect the correctness at
all.
Among the classic local problems, as studied in the LOCAL model (see~\cite{Peleg00} for the original definition and~\cite{HirvonenS20} for   a recent book), there are mainly coloring problems and packing/covering problems. 
The coloring problems (with a fixed number of colors) are not meaningful in our context: an edge removal can only help. 
But the packing/covering problems are all interesting, thus we widen the scope to cover three classic problems in this paper:
maximal independent set (MIS) as before, but also maximal matching (MM) and minimal dominating set (MDS).



To help the reader grasp some intuition on our approach, let us illustrate the $1$-robustness for the maximal
matching, \emph{i.e.} a set of edges that do not share vertices
and that is maximal in the sense that no edge can be added. 
To be $1$-robust, a matching must still be maximal after the removal of \emph{one arbitrary} edge that does not disconnect the graph. 
Let us go over various configurations illustrated in Figure~\ref{fig:MM} (the matched edges are bold ones).

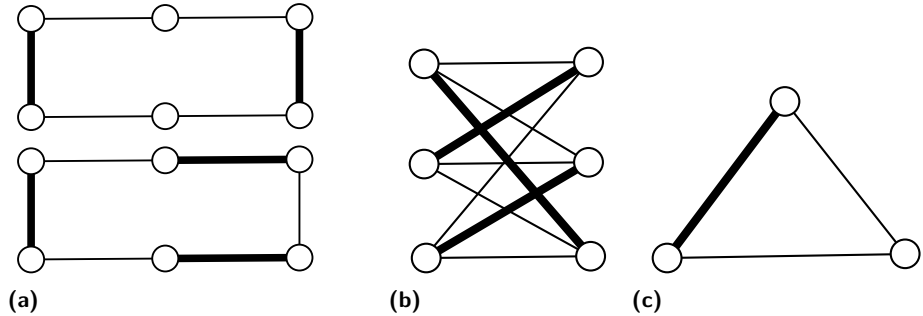
\begin{figure}[!h]
  \centering
  \begin{minipage}[b]{.35\textwidth}
  
	\scalebox{0.95}{\tikzset{every picture/.style={line width=0.75pt}} 

\begin{tikzpicture}[x=0.75pt,y=0.75pt,yscale=-1,xscale=1]

\draw [line width=3]    (106.85,124.85) -- (106.85,176.85) ;
\draw [line width=3]    (247.85,123.85) -- (247.85,175.85) ;
\draw    (106.85,124.85) -- (247.85,123.85) ;
\draw  [fill={rgb, 255:red, 255; green, 255; blue, 255 }  ,fill opacity=1 ] (100,124.85) .. controls (100,121.07) and (103.07,118) .. (106.85,118) .. controls (110.63,118) and (113.7,121.07) .. (113.7,124.85) .. controls (113.7,128.63) and (110.63,131.7) .. (106.85,131.7) .. controls (103.07,131.7) and (100,128.63) .. (100,124.85) -- cycle ;
\draw  [fill={rgb, 255:red, 255; green, 255; blue, 255 }  ,fill opacity=1 ] (170.5,124.35) .. controls (170.5,120.57) and (173.57,117.5) .. (177.35,117.5) .. controls (181.13,117.5) and (184.2,120.57) .. (184.2,124.35) .. controls (184.2,128.13) and (181.13,131.2) .. (177.35,131.2) .. controls (173.57,131.2) and (170.5,128.13) .. (170.5,124.35) -- cycle ;
\draw  [fill={rgb, 255:red, 255; green, 255; blue, 255 }  ,fill opacity=1 ] (241,123.85) .. controls (241,120.07) and (244.07,117) .. (247.85,117) .. controls (251.63,117) and (254.7,120.07) .. (254.7,123.85) .. controls (254.7,127.63) and (251.63,130.7) .. (247.85,130.7) .. controls (244.07,130.7) and (241,127.63) .. (241,123.85) -- cycle ;
\draw    (106.85,176.85) -- (247.85,175.85) ;
\draw  [fill={rgb, 255:red, 255; green, 255; blue, 255 }  ,fill opacity=1 ] (100,176.85) .. controls (100,173.07) and (103.07,170) .. (106.85,170) .. controls (110.63,170) and (113.7,173.07) .. (113.7,176.85) .. controls (113.7,180.63) and (110.63,183.7) .. (106.85,183.7) .. controls (103.07,183.7) and (100,180.63) .. (100,176.85) -- cycle ;
\draw  [fill={rgb, 255:red, 255; green, 255; blue, 255 }  ,fill opacity=1 ] (170.5,176.35) .. controls (170.5,172.57) and (173.57,169.5) .. (177.35,169.5) .. controls (181.13,169.5) and (184.2,172.57) .. (184.2,176.35) .. controls (184.2,180.13) and (181.13,183.2) .. (177.35,183.2) .. controls (173.57,183.2) and (170.5,180.13) .. (170.5,176.35) -- cycle ;
\draw  [fill={rgb, 255:red, 255; green, 255; blue, 255 }  ,fill opacity=1 ] (241,175.85) .. controls (241,172.07) and (244.07,169) .. (247.85,169) .. controls (251.63,169) and (254.7,172.07) .. (254.7,175.85) .. controls (254.7,179.63) and (251.63,182.7) .. (247.85,182.7) .. controls (244.07,182.7) and (241,179.63) .. (241,175.85) -- cycle ;

\end{tikzpicture}}\\[-0.3cm]
	
	\scalebox{0.95}{\tikzset{every picture/.style={line width=0.75pt}} 

\begin{tikzpicture}[x=0.75pt,y=0.75pt,yscale=-1,xscale=1]

\draw [line width=3]    (177.35,176.35) -- (247.85,175.85) ;
\draw [line width=3]    (177.35,124.35) -- (247.85,123.85) ;
\draw [line width=3]    (106.85,124.85) -- (106.85,176.85) ;
\draw [line width=0.75]    (247.85,123.85) -- (247.85,175.85) ;
\draw    (106.85,124.85) -- (247.85,123.85) ;
\draw  [fill={rgb, 255:red, 255; green, 255; blue, 255 }  ,fill opacity=1 ] (100,124.85) .. controls (100,121.07) and (103.07,118) .. (106.85,118) .. controls (110.63,118) and (113.7,121.07) .. (113.7,124.85) .. controls (113.7,128.63) and (110.63,131.7) .. (106.85,131.7) .. controls (103.07,131.7) and (100,128.63) .. (100,124.85) -- cycle ;
\draw  [fill={rgb, 255:red, 255; green, 255; blue, 255 }  ,fill opacity=1 ] (170.5,124.35) .. controls (170.5,120.57) and (173.57,117.5) .. (177.35,117.5) .. controls (181.13,117.5) and (184.2,120.57) .. (184.2,124.35) .. controls (184.2,128.13) and (181.13,131.2) .. (177.35,131.2) .. controls (173.57,131.2) and (170.5,128.13) .. (170.5,124.35) -- cycle ;
\draw  [fill={rgb, 255:red, 255; green, 255; blue, 255 }  ,fill opacity=1 ] (241,123.85) .. controls (241,120.07) and (244.07,117) .. (247.85,117) .. controls (251.63,117) and (254.7,120.07) .. (254.7,123.85) .. controls (254.7,127.63) and (251.63,130.7) .. (247.85,130.7) .. controls (244.07,130.7) and (241,127.63) .. (241,123.85) -- cycle ;
\draw    (106.85,176.85) -- (247.85,175.85) ;
\draw  [fill={rgb, 255:red, 255; green, 255; blue, 255 }  ,fill opacity=1 ] (100,176.85) .. controls (100,173.07) and (103.07,170) .. (106.85,170) .. controls (110.63,170) and (113.7,173.07) .. (113.7,176.85) .. controls (113.7,180.63) and (110.63,183.7) .. (106.85,183.7) .. controls (103.07,183.7) and (100,180.63) .. (100,176.85) -- cycle ;
\draw  [fill={rgb, 255:red, 255; green, 255; blue, 255 }  ,fill opacity=1 ] (170.5,176.35) .. controls (170.5,172.57) and (173.57,169.5) .. (177.35,169.5) .. controls (181.13,169.5) and (184.2,172.57) .. (184.2,176.35) .. controls (184.2,180.13) and (181.13,183.2) .. (177.35,183.2) .. controls (173.57,183.2) and (170.5,180.13) .. (170.5,176.35) -- cycle ;
\draw  [fill={rgb, 255:red, 255; green, 255; blue, 255 }  ,fill opacity=1 ] (241,175.85) .. controls (241,172.07) and (244.07,169) .. (247.85,169) .. controls (251.63,169) and (254.7,172.07) .. (254.7,175.85) .. controls (254.7,179.63) and (251.63,182.7) .. (247.85,182.7) .. controls (244.07,182.7) and (241,179.63) .. (241,175.85) -- cycle ;

\end{tikzpicture}}
  \subcaption{}\label{fig:mm-a}
  \end{minipage}
  \begin{minipage}[b]{.22\textwidth}
  \scalebox{1.05}{
\tikzset{every picture/.style={line width=0.75pt}} 

\begin{tikzpicture}[x=0.75pt,y=0.75pt,yscale=-1,xscale=1]
\draw    (171.85,184.85) -- (250,91) ;

\draw    (171.85,139.85) -- (251,184) ;
\draw    (171.85,91.85) -- (250,139) ;
\draw [line width=3]    (172.85,184.85) -- (250,139) ;
\draw [line width=3]    (171.85,139.85) -- (250,91) ;
\draw [line width=3]    (171.85,91.85) -- (251,184) ;
\draw    (171.85,91.85) -- (250,91) ;
\draw  [fill={rgb, 255:red, 255; green, 255; blue, 255 }  ,fill opacity=1 ] (165,91.85) .. controls (165,88.07) and (168.07,85) .. (171.85,85) .. controls (175.63,85) and (178.7,88.07) .. (178.7,91.85) .. controls (178.7,95.63) and (175.63,98.7) .. (171.85,98.7) .. controls (168.07,98.7) and (165,95.63) .. (165,91.85) -- cycle ;
\draw  [fill={rgb, 255:red, 255; green, 255; blue, 255 }  ,fill opacity=1 ] (243.15,91) .. controls (243.15,87.22) and (246.22,84.15) .. (250,84.15) .. controls (253.78,84.15) and (256.85,87.22) .. (256.85,91) .. controls (256.85,94.78) and (253.78,97.85) .. (250,97.85) .. controls (246.22,97.85) and (243.15,94.78) .. (243.15,91) -- cycle ;
\draw    (171.85,139.85) -- (250,139) ;
\draw  [fill={rgb, 255:red, 255; green, 255; blue, 255 }  ,fill opacity=1 ] (165,139.85) .. controls (165,136.07) and (168.07,133) .. (171.85,133) .. controls (175.63,133) and (178.7,136.07) .. (178.7,139.85) .. controls (178.7,143.63) and (175.63,146.7) .. (171.85,146.7) .. controls (168.07,146.7) and (165,143.63) .. (165,139.85) -- cycle ;
\draw  [fill={rgb, 255:red, 255; green, 255; blue, 255 }  ,fill opacity=1 ] (243.15,139) .. controls (243.15,135.22) and (246.22,132.15) .. (250,132.15) .. controls (253.78,132.15) and (256.85,135.22) .. (256.85,139) .. controls (256.85,142.78) and (253.78,145.85) .. (250,145.85) .. controls (246.22,145.85) and (243.15,142.78) .. (243.15,139) -- cycle ;
\draw    (172.85,184.85) -- (251,184) ;
\draw  [fill={rgb, 255:red, 255; green, 255; blue, 255 }  ,fill opacity=1 ] (166,184.85) .. controls (166,181.07) and (169.07,178) .. (172.85,178) .. controls (176.63,178) and (179.7,181.07) .. (179.7,184.85) .. controls (179.7,188.63) and (176.63,191.7) .. (172.85,191.7) .. controls (169.07,191.7) and (166,188.63) .. (166,184.85) -- cycle ;
\draw  [fill={rgb, 255:red, 255; green, 255; blue, 255 }  ,fill opacity=1 ] (244.15,184) .. controls (244.15,180.22) and (247.22,177.15) .. (251,177.15) .. controls (254.78,177.15) and (257.85,180.22) .. (257.85,184) .. controls (257.85,187.78) and (254.78,190.85) .. (251,190.85) .. controls (247.22,190.85) and (244.15,187.78) .. (244.15,184) -- cycle ;

\end{tikzpicture}}
    \subcaption{}\label{fig:mm-b}
  \end{minipage}
  \begin{minipage}[b]{.25\textwidth}
  \scalebox{1.05}{
           \tikzset{every picture/.style={line width=0.75pt}} 

\begin{tikzpicture}[x=0.75pt,y=0.75pt,yscale=-1,xscale=1]

\draw    (396,159) -- (516,157) ;
\draw [line width=3]    (396,159) -- (452,84) ;
\draw    (452,84) -- (509.15,157) ;
\draw  [fill={rgb, 255:red, 255; green, 255; blue, 255 }  ,fill opacity=1 ] (445.15,84) .. controls (445.15,80.22) and (448.22,77.15) .. (452,77.15) .. controls (455.78,77.15) and (458.85,80.22) .. (458.85,84) .. controls (458.85,87.78) and (455.78,90.85) .. (452,90.85) .. controls (448.22,90.85) and (445.15,87.78) .. (445.15,84) -- cycle ;
\draw  [fill={rgb, 255:red, 255; green, 255; blue, 255 }  ,fill opacity=1 ] (389.15,159) .. controls (389.15,155.22) and (392.22,152.15) .. (396,152.15) .. controls (399.78,152.15) and (402.85,155.22) .. (402.85,159) .. controls (402.85,162.78) and (399.78,165.85) .. (396,165.85) .. controls (392.22,165.85) and (389.15,162.78) .. (389.15,159) -- cycle ;
\draw  [fill={rgb, 255:red, 255; green, 255; blue, 255 }  ,fill opacity=1 ] (502.3,157) .. controls (502.3,153.22) and (505.37,150.15) .. (509.15,150.15) .. controls (512.93,150.15) and (516,153.22) .. (516,157) .. controls (516,160.78) and (512.93,163.85) .. (509.15,163.85) .. controls (505.37,163.85) and (502.3,160.78) .. (502.3,157) -- cycle ;

\end{tikzpicture}}
     \subcaption{}\label{fig:mm-c}
  \end{minipage}
  \caption{\label{fig:MM} Three examples of MMs in various graphs.}
\end{figure}

For the two graphs in Figure~\ref{fig:mm-a}, that are cycles of 6 vertices, we can observe that two instances of maximal matching can have different behaviors. 
Indeed, in the top one, if we remove one matched edge, we are left
with a matching that is not maximal in the new graph: the two edges adjacent to the removed one could be added.  By
contrast, in the bottom graph, any edge removal leaves a graph that is still a maximal matching. 
Now, in the graph of Figure~\ref{fig:mm-b}, a complete balanced bipartite graph, all the maximal matchings are identical up to isomorphism. 
After one arbitrary edge removal, we are left with a graph where no new edge can be matched. Therefore in this graph, any matching is robust to one edge removal. Note that this is not true for any number of edge removals, illustrating the fact that $k$-robustness and robustness are not equivalent.
Finally, in Figure~\ref{fig:mm-c}, all the maximal matchings consists of only one edge, and they are not robust to an
edge removal.  Indeed, after the matched edge is removed, one can choose any of the two remaining ones. 

To summarize, Figure~\ref{fig:MM} illustrates the effect of $1$-robustness in three different cases: one where
\textit{some} matchings are $1$-robust, one where \textit{all} matchings are $1$-robust, and one where \textit{no} matching is
$1$-robust.   

\subsection{Our results}

Our first contribution is to introduce the fine-grained version of robustness in Section~\ref{sec:model}. 
After that, every technical section of this paper is devoted to provide some answer to the fine-grained version of one of the two questions highlighted above (existential \textit{vs.} universal) for one of the problems we study. 
Our focus is actually in understanding how do the different settings compare, in terms of both problems and number of removable edges. 

Let us start with the universal question. Here, we prove that the three problems have three different behaviors.

For minimal dominating set, the class of the graphs for which any solution is $k$-robust is exactly the same for every $k$ (a class that already appeared in \cite{CasteigtsDPR20} under the name of \emph{sputnik graphs}) as proved in Section~\ref{sec:MDS}.
 
For maximal matching, the case of $k=1$, which we used previously as an example, is special and draws an interesting connection with perfect matchings,  but then the class is identical for every $k\geq 2$. These results are presented in Section~\ref{sec:MM}.
 
Finally, for maximal independent set, we show in Section~\ref{sec:MIS} that there is a strict hierarchy: the class for $k$ edge removals is strictly smaller than the one for $k-1$. For this case, we do not pinpoint the exact characterization, but give some additional structural results on the classes.  

The existential question is much more challenging. Section~\ref{sec:NPH} presents some preliminary results on the study of this question.
For maximal independent set, we show that for any~$k$, deciding whether a graph has a maximal independent set that is robust to $k$ edge removals is NP-hard. This is the first NP-hardness result for this type of question. 

\section{Model, definitions, and basic properties}
\label{sec:model}

In the paper, except when stated otherwise, the graph is named $G$, the vertex set $V$ and the edge set $E$.

\subsection{Robustness and graph problems}

The key notion of this paper is the one of $k$-robustness.

\begin{definition}
Given a graph problem and a graph, a solution is \emph{$k$-robust} if after the removal of \emph{at most $k$} edges, either the graph is disconnected, or the solution is still valid. 
\end{definition}

Note that $k$-robustness is about removing at most $k$ edges, not exactly $k$ edges.

We will abuse notation and write $\infty$-robust when mentioning the notion of robustness from \cite{CasteigtsDPR20}, with an unbounded number of removals. Hence $k$ is a parameter in $\mathbb{N}\cup \infty$.

\begin{notation}
We define $\UU^k_{P}$ and $\mathcal{E}^k_{P}$ the following way:
\begin{itemize}
\item Let $\UU^k_{P}$ be the class of graphs such that any solution to the problem $P$ is $k$-robust.
\item Let $\mathcal{E}^k_{P}$ be the class of graphs such that there exists a solution to the problem $P$ that is $k$-robust
\end{itemize}
\end{notation}

Note that to easily incorporate the parameter $k$, we decided to not follow the exact same  notations as in~\cite{CasteigtsDPR20}.

\paragraph*{Graph problems.}
We consider three graph problems:
\begin{enumerate}
\item Minimal dominating set (MDS): Select a minimal set of vertices such that every vertex of the graph is either in the set or has a neighbor in the set.
\item Maximal matching (MM): Select a maximal set of edges such that no two selected edges share endpoint.
\item Maximal independent set (MIS): Select a maximal set of vertices such that no two selected vertices share an edge.
\end{enumerate}

A \emph{perfect matching} is a matching where every vertex is matched.
We will also use the notion of \emph{$k$-dominating set}, which is a set of selected vertices such that every vertex is either selected or is adjacent to two selected vertices.
Note that $k$-dominating set sometimes refer to another notion related to the distance to the selected vertices, but this is not our definition.   

\paragraph*{The case of robust maximal matching.}

For maximal matching, the definition of robustness may vary. The definition we take is the following.
A maximal matching $M$ of a graph $G$ is $k$-robust if after removing any set of at most $k$ edges such that the graph $G$ is still connected, what remains of $M$ is a maximal matching of what remains of $G$.

\subsection{Graph notions}
\label{subsec:graph-notions}

We list a few graph theory definitions that we will need.

\begin{definition}
The \emph{neighborhood} of a node $v$, denoted $N(v)$, is the set of nodes that are adjacent to $v$. 
The \emph{closed neighborhood} of a node $v$, denoted $N[v]$, is the neighborhood of $v$, plus $v$ itself.
\end{definition}

\begin{definition}
A graph is $t$-(edge)-connected if, after the removal of any set of $(t-1)$ edges, the graph is still connected. A $t$-(edge)-connected component is a subgraph that is $t$-(edge)-connected.
\end{definition}

In the following we are only interested in edge connectivity thus we will simply write \emph{$t$-connectivity} to refer to \emph{$t$-edge-connectivity}. 
In our proofs, we will use the following easy observation multiple times : in a 2-connected graph every vertex belongs to a cycle.

\begin{definition}
In a connected graph, a \emph{bridge} is an edge whose removal disconnects the graph. 
\end{definition}

\begin{definition}\label{def:join}
Given two graphs $G$ and $H$, the \emph{join} of these two graphs, $join(G,H)$, is the graph made by taking the union of $G$ and $H$, and adding all the possible edges $(u,v)$, with $u\in G$ and $v\in H$. See Figure~\ref{fig:join}.
\end{definition}

\begin{definition}
A \emph{sputnik graph} (\cite{CasteigtsDPR20}) is a graph where every node that is part of a cycle has an \emph{antenna}, that is a neighbor with degree 1. See Figure~\ref{fig:sputnik}.
\end{definition}

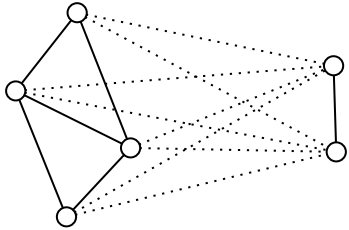
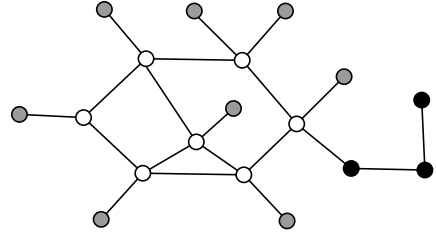
\begin{figure}[!h]
\begin{center}
\begin{subfigure}[t]{0.45 \textwidth}
\tikzset{every picture/.style={line width=0.75pt}} 

\begin{tikzpicture}[x=0.75pt,y=0.75pt,yscale=-1,xscale=1]

\draw    (304.13,111.79) -- (305.46,155.13) ;
\draw    (202.79,153.13) -- (170.79,187.79) ;
\draw    (145.46,124.46) -- (202.79,153.13) ;
\draw    (145.46,124.46) -- (170.79,187.79) ;
\draw    (176.13,85.13) -- (145.46,124.46) ;
\draw    (176.13,85.13) -- (202.79,153.13) ;
\draw  [dash pattern={on 0.84pt off 2.51pt}]  (176.13,85.13) -- (304.13,111.79) ;
\draw  [dash pattern={on 0.84pt off 2.51pt}]  (176.13,85.13) -- (305.46,155.13) ;
\draw  [dash pattern={on 0.84pt off 2.51pt}]  (202.79,153.13) -- (304.13,111.79) ;
\draw  [dash pattern={on 0.84pt off 2.51pt}]  (202.79,153.13) -- (305.46,155.13) ;
\draw  [dash pattern={on 0.84pt off 2.51pt}]  (170.79,187.79) -- (304.13,111.79) ;
\draw  [dash pattern={on 0.84pt off 2.51pt}]  (170.79,187.79) -- (305.46,155.13) ;
\draw  [dash pattern={on 0.84pt off 2.51pt}]  (145.46,124.46) -- (304.13,111.79) ;
\draw  [dash pattern={on 0.84pt off 2.51pt}]  (145.46,124.46) -- (305.46,155.13) ;
\draw  [fill={rgb, 255:red, 255; green, 255; blue, 255 }  ,fill opacity=1 ] (171.33,85.14) .. controls (171.33,82.49) and (173.47,80.34) .. (176.11,80.33) .. controls (178.76,80.33) and (180.91,82.47) .. (180.92,85.11) .. controls (180.93,87.76) and (178.79,89.91) .. (176.14,89.92) .. controls (173.49,89.93) and (171.34,87.79) .. (171.33,85.14) -- cycle ;
\draw  [fill={rgb, 255:red, 255; green, 255; blue, 255 }  ,fill opacity=1 ] (198,153.14) .. controls (197.99,150.49) and (200.13,148.34) .. (202.78,148.33) .. controls (205.43,148.33) and (207.58,150.47) .. (207.59,153.11) .. controls (207.6,155.76) and (205.46,157.91) .. (202.81,157.92) .. controls (200.16,157.93) and (198.01,155.79) .. (198,153.14) -- cycle ;
\draw  [fill={rgb, 255:red, 255; green, 255; blue, 255 }  ,fill opacity=1 ] (140.67,124.48) .. controls (140.66,121.83) and (142.8,119.67) .. (145.45,119.67) .. controls (148.09,119.66) and (150.25,121.8) .. (150.26,124.45) .. controls (150.26,127.09) and (148.12,129.25) .. (145.48,129.26) .. controls (142.83,129.26) and (140.67,127.12) .. (140.67,124.48) -- cycle ;
\draw  [fill={rgb, 255:red, 255; green, 255; blue, 255 }  ,fill opacity=1 ] (166,187.81) .. controls (165.99,185.16) and (168.13,183.01) .. (170.78,183) .. controls (173.43,182.99) and (175.58,185.13) .. (175.59,187.78) .. controls (175.6,190.43) and (173.46,192.58) .. (170.81,192.59) .. controls (168.16,192.6) and (166.01,190.46) .. (166,187.81) -- cycle ;
\draw  [fill={rgb, 255:red, 255; green, 255; blue, 255 }  ,fill opacity=1 ] (299.33,111.81) .. controls (299.33,109.16) and (301.47,107.01) .. (304.11,107) .. controls (306.76,106.99) and (308.91,109.13) .. (308.92,111.78) .. controls (308.93,114.43) and (306.79,116.58) .. (304.14,116.59) .. controls (301.49,116.6) and (299.34,114.46) .. (299.33,111.81) -- cycle ;
\draw  [fill={rgb, 255:red, 255; green, 255; blue, 255 }  ,fill opacity=1 ] (300.67,155.14) .. controls (300.66,152.49) and (302.8,150.34) .. (305.45,150.33) .. controls (308.09,150.33) and (310.25,152.47) .. (310.26,155.11) .. controls (310.26,157.76) and (308.12,159.91) .. (305.48,159.92) .. controls (302.83,159.93) and (300.67,157.79) .. (300.67,155.14) -- cycle ;

\end{tikzpicture}
\vspace{0.2cm}
\caption{\label{fig:join} The join of two graphs: the black edges are the original edges, the doted edges are the one added by the join operation.}
\end{subfigure}
\hspace{0.5cm}
\begin{subfigure}[t]{0.45 \textwidth}
\scalebox{0.8}{
\tikzset{every picture/.style={line width=0.75pt}} 

\begin{tikzpicture}[x=0.75pt,y=0.75pt,yscale=-1,xscale=1]

\draw    (273.13,95.79) -- (300.13,64.79) ;
\draw    (273.13,95.79) -- (307.13,135.79) ;
\draw    (307.13,135.79) -- (336.7,106.1) ;
\draw    (307.13,135.79) -- (341.13,163.79) ;
\draw    (274.13,167.79) -- (307.13,135.79) ;
\draw    (274.13,167.79) -- (301.13,196.79) ;
\draw    (211.13,166.79) -- (274.13,167.79) ;
\draw    (244.49,147.09) -- (274.13,167.79) ;
\draw    (244.49,147.09) -- (267.7,126.1) ;
\draw    (211.13,166.79) -- (244.49,147.09) ;
\draw    (185.13,195.79) -- (211.13,166.79) ;
\draw    (174.13,131.79) -- (211.13,166.79) ;

\draw    (242,65) -- (272,95) ;
\draw    (134.13,129.79) -- (174.13,131.79) ;
\draw    (174.13,131.79) -- (213.13,94.79) ;
\draw    (187.13,63.79) -- (213.13,94.79) ;
\draw    (213.13,94.79) -- (273.13,95.79) ;
\draw    (213.14,99.59) -- (244.49,147.09) ;
\draw  [fill={rgb, 255:red, 255; green, 255; blue, 255 }  ,fill opacity=1 ] (208.33,94.81) .. controls (208.33,92.16) and (210.47,90.01) .. (213.11,90) .. controls (215.76,89.99) and (217.91,92.13) .. (217.92,94.78) .. controls (217.93,97.43) and (215.79,99.58) .. (213.14,99.59) .. controls (210.49,99.6) and (208.34,97.46) .. (208.33,94.81) -- cycle ;
\draw  [fill={rgb, 255:red, 255; green, 255; blue, 255 }  ,fill opacity=1 ] (268.33,95.81) .. controls (268.33,93.16) and (270.47,91.01) .. (273.11,91) .. controls (275.76,90.99) and (277.91,93.13) .. (277.92,95.78) .. controls (277.93,98.43) and (275.79,100.58) .. (273.14,100.59) .. controls (270.49,100.6) and (268.34,98.46) .. (268.33,95.81) -- cycle ;
\draw  [fill={rgb, 255:red, 255; green, 255; blue, 255 }  ,fill opacity=1 ] (302.33,135.81) .. controls (302.33,133.16) and (304.47,131.01) .. (307.11,131) .. controls (309.76,130.99) and (311.91,133.13) .. (311.92,135.78) .. controls (311.93,138.43) and (309.79,140.58) .. (307.14,140.59) .. controls (304.49,140.6) and (302.34,138.46) .. (302.33,135.81) -- cycle ;
\draw  [fill={rgb, 255:red, 255; green, 255; blue, 255 }  ,fill opacity=1 ] (269.33,167.81) .. controls (269.33,165.16) and (271.47,163.01) .. (274.11,163) .. controls (276.76,162.99) and (278.91,165.13) .. (278.92,167.78) .. controls (278.93,170.43) and (276.79,172.58) .. (274.14,172.59) .. controls (271.49,172.6) and (269.34,170.46) .. (269.33,167.81) -- cycle ;
\draw  [fill={rgb, 255:red, 255; green, 255; blue, 255 }  ,fill opacity=1 ] (206.33,166.81) .. controls (206.33,164.16) and (208.47,162.01) .. (211.11,162) .. controls (213.76,161.99) and (215.91,164.13) .. (215.92,166.78) .. controls (215.93,169.43) and (213.79,171.58) .. (211.14,171.59) .. controls (208.49,171.6) and (206.34,169.46) .. (206.33,166.81) -- cycle ;
\draw  [fill={rgb, 255:red, 255; green, 255; blue, 255 }  ,fill opacity=1 ] (169.33,131.81) .. controls (169.33,129.16) and (171.47,127.01) .. (174.11,127) .. controls (176.76,126.99) and (178.91,129.13) .. (178.92,131.78) .. controls (178.93,134.43) and (176.79,136.58) .. (174.14,136.59) .. controls (171.49,136.6) and (169.34,134.46) .. (169.33,131.81) -- cycle ;
\draw  [fill={rgb, 255:red, 255; green, 255; blue, 255 }  ,fill opacity=1 ] (239.7,147.1) .. controls (239.69,144.45) and (241.83,142.3) .. (244.48,142.29) .. controls (247.13,142.28) and (249.28,144.42) .. (249.29,147.07) .. controls (249.3,149.72) and (247.16,151.87) .. (244.51,151.88) .. controls (241.86,151.89) and (239.71,149.75) .. (239.7,147.1) -- cycle ;
\draw  [fill={rgb, 255:red, 155; green, 155; blue, 155 }  ,fill opacity=1 ] (182.33,63.81) .. controls (182.33,61.16) and (184.47,59.01) .. (187.11,59) .. controls (189.76,58.99) and (191.91,61.13) .. (191.92,63.78) .. controls (191.93,66.43) and (189.79,68.58) .. (187.14,68.59) .. controls (184.49,68.6) and (182.34,66.46) .. (182.33,63.81) -- cycle ;
\draw  [fill={rgb, 255:red, 155; green, 155; blue, 155 }  ,fill opacity=1 ] (295.33,64.81) .. controls (295.33,62.16) and (297.47,60.01) .. (300.11,60) .. controls (302.76,59.99) and (304.91,62.13) .. (304.92,64.78) .. controls (304.93,67.43) and (302.79,69.58) .. (300.14,69.59) .. controls (297.49,69.6) and (295.34,67.46) .. (295.33,64.81) -- cycle ;

\draw  [shift={(-1.5cm,0cm)}, fill={rgb, 255:red, 155; green, 155; blue, 155 }  ,fill opacity=1 ] (295.33,64.81) .. controls (295.33,62.16) and (297.47,60.01) .. (300.11,60) .. controls (302.76,59.99) and (304.91,62.13) .. (304.92,64.78) .. controls (304.93,67.43) and (302.79,69.58) .. (300.14,69.59) .. controls (297.49,69.6) and (295.34,67.46) .. (295.33,64.81) -- cycle ;

\draw  [fill={rgb, 255:red, 155; green, 155; blue, 155 }  ,fill opacity=1 ] (129.33,129.81) .. controls (129.33,127.16) and (131.47,125.01) .. (134.11,125) .. controls (136.76,124.99) and (138.91,127.13) .. (138.92,129.78) .. controls (138.93,132.43) and (136.79,134.58) .. (134.14,134.59) .. controls (131.49,134.6) and (129.34,132.46) .. (129.33,129.81) -- cycle ;
\draw  [fill={rgb, 255:red, 155; green, 155; blue, 155 }  ,fill opacity=1 ] (180.33,195.81) .. controls (180.33,193.16) and (182.47,191.01) .. (185.11,191) .. controls (187.76,190.99) and (189.91,193.13) .. (189.92,195.78) .. controls (189.93,198.43) and (187.79,200.58) .. (185.14,200.59) .. controls (182.49,200.6) and (180.34,198.46) .. (180.33,195.81) -- cycle ;
\draw  [fill={rgb, 255:red, 155; green, 155; blue, 155 }  ,fill opacity=1 ] (296.33,196.81) .. controls (296.33,194.16) and (298.47,192.01) .. (301.11,192) .. controls (303.76,191.99) and (305.91,194.13) .. (305.92,196.78) .. controls (305.93,199.43) and (303.79,201.58) .. (301.14,201.59) .. controls (298.49,201.6) and (296.34,199.46) .. (296.33,196.81) -- cycle ;
\draw  [fill={rgb, 255:red, 155; green, 155; blue, 155 }  ,fill opacity=1 ] (331.91,106.11) .. controls (331.9,103.47) and (334.04,101.31) .. (336.69,101.31) .. controls (339.33,101.3) and (341.49,103.44) .. (341.49,106.09) .. controls (341.5,108.73) and (339.36,110.89) .. (336.71,110.89) .. controls (334.07,110.9) and (331.91,108.76) .. (331.91,106.11) -- cycle ;
\draw  [fill={rgb, 255:red, 0; green, 0; blue, 0 }  ,fill opacity=1 ] (336.33,163.81) .. controls (336.33,161.16) and (338.47,159.01) .. (341.11,159) .. controls (343.76,158.99) and (345.91,161.13) .. (345.92,163.78) .. controls (345.93,166.43) and (343.79,168.58) .. (341.14,168.59) .. controls (338.49,168.6) and (336.34,166.46) .. (336.33,163.81) -- cycle ;
\draw  [fill={rgb, 255:red, 0; green, 0; blue, 0 }  ,fill opacity=1 ] (382.33,164.81) .. controls (382.33,162.16) and (384.47,160.01) .. (387.11,160) .. controls (389.76,159.99) and (391.91,162.13) .. (391.92,164.78) .. controls (391.93,167.43) and (389.79,169.58) .. (387.14,169.59) .. controls (384.49,169.6) and (382.34,167.46) .. (382.33,164.81) -- cycle ;
\draw    (341.13,163.79) -- (387.13,164.79) ;
\draw  [fill={rgb, 255:red, 155; green, 155; blue, 155 }  ,fill opacity=1 ] (262.91,126.11) .. controls (262.9,123.47) and (265.04,121.31) .. (267.69,121.31) .. controls (270.33,121.3) and (272.49,123.44) .. (272.49,126.09) .. controls (272.5,128.73) and (270.36,130.89) .. (267.71,130.89) .. controls (265.07,130.9) and (262.91,128.76) .. (262.91,126.11) -- cycle ;
\draw  [fill={rgb, 255:red, 0; green, 0; blue, 0 }  ,fill opacity=1 ] (380.33,120.81) .. controls (380.33,118.16) and (382.47,116.01) .. (385.11,116) .. controls (387.76,115.99) and (389.91,118.13) .. (389.92,120.78) .. controls (389.93,123.43) and (387.79,125.58) .. (385.14,125.59) .. controls (382.49,125.6) and (380.34,123.46) .. (380.33,120.81) -- cycle ;
\draw    (385.13,120.79) -- (387.13,164.79) ;

\end{tikzpicture}}
\vspace{0.2cm}
\caption{\label{fig:sputnik}A sputnik graph. The white vertices are part a cycles, the grey vertices are their antennas, and the black vertices do not belong to any cycle, nor are antennas.}
\end{subfigure}
\end{center}
\caption{
Illustration of the definitions of Subsection~\ref{subsec:graph-notions}.
}
\end{figure}

\subsection{Basic properties}

The following properties follow from the definitions.

\begin{property}\label{prop:basic-inclusions}
For any problem $P$, for any $k$, $\UU^{k+1}_{P}\subseteq \UU^{k}_{P}$ and $\mathcal{E}^{k+1}_{P}\subseteq \mathcal{E}^{k}_{P}$.
\end{property}

In particular, $\UU^{\infty}_{P}\subseteq \UU^{k}_{P} \subseteq \UU^{1}_{P}$ and $\mathcal{E}^{\infty}_{P}\subseteq \mathcal{E}^{k}_{P} \subseteq \mathcal{E}^{1}_{P}$, for all $k$.

\begin{property}\label{prop:connectivity}
If a graph is \emph{$(k+1)$-connected} then a solution is $k$-robust if and only if after the removal of any set of $k$ edges the solution is still correct.
\end{property}


\section{Minimal dominating set}
\label{sec:MDS}

\begin{theorem}
\label{thm:MDS}
For all $k$ in $\mathbb{N}\cup \infty$, $\UU_{MDS}^k$ is the set of sputnik graphs.
\end{theorem}

\begin{proof}
We know from~\cite{CasteigtsDPR20} that the theorem holds for $k=\infty$.
Hence, thanks to Property~\ref{prop:basic-inclusions}, it is sufficient to prove that the theorem is true for $k=1$.
For the sake of contradiction, consider a graph~$G$ in $\UU_{MDS}^1$ that is not a sputnik graph.
Then there is a node $u$ that belongs to a cycle, and that has no antenna. Let $S$ be the closed neighborhood of $u$, $S=N[u]$. 
We say that a node of $S$, different from $u$, is an \emph{inside node} if it is only connected to nodes in~$S$.
We now consider two cases depending on whether there is an inside node or not. See Figure~\ref{fig:proof-MDS}.

\begin{enumerate}
\item Suppose there exists an inside node $v$. Note that $v$ has at least one neighbor different from $u$ because otherwise it would be an antenna.
Let the set $W$ be the closed neighborhood of $v$, except~$u$. The set $D=V\setminus W$ is a dominating set of the graph, because all the nodes either  belong to $D$ or are neighbors of $u$ (which belongs to $D$).
Now, we transform $D$ into a \emph{minimal} dominating set greedily: we remove nodes from $D$ in an arbitrary order, until no more nodes can be removed without making $D$ non-dominating. 
We claim that this minimal dominating set is not 1-robust.
Indeed, if we remove the edge $(u,v)$, $v$ is not covered any more (none of its current neighbors belongs to $D$), and the graph is still connected (because $v$ has a neighbor different from $u$). 

\item Suppose there is no inside vertex. 
Let $a$ be a neighbor of $u$ in the cycle. Let $W$ be the set $S\setminus a$. 
Again we claim that $V \setminus W$ is a dominating set. 
Indeed, because there is no inside node, every node in $S$ different from $u$ is covered by node outside $W$, and $u$ is covered by~$a$, which belongs to $V\setminus W$.
As before we can make this set an MDS by removing nodes greedily, and again we claim it is not 1-robust.  
Indeed, if we remove the edge $(u,a)$, we do not disconnect the graph (because of the cycle containing $u$), and $u$ is left uncovered.
\end{enumerate}
\end{proof}

\begin{figure}[!h]
\begin{center}
\begin{tabular}{cc}
\tikzset{every picture/.style={line width=0.75pt}} 

\begin{tikzpicture}[x=0.75pt,y=0.75pt,yscale=-1,xscale=1]

\draw  [dash pattern={on 0.84pt off 2.51pt}]  (268.5,53) -- (301,41) ;
\draw  [dash pattern={on 0.84pt off 2.51pt}]  (268.5,53) -- (307,67) ;
\draw  [dash pattern={on 4.5pt off 4.5pt}] (140.7,48.1) .. controls (160.7,38.1) and (218.3,37.9) .. (230,50) .. controls (241.7,62.1) and (244.7,113.1) .. (229.7,129.1) .. controls (214.7,145.1) and (146.7,145.1) .. (135.7,129.1) .. controls (124.7,113.1) and (120.7,58.1) .. (140.7,48.1) -- cycle ;
\draw  [color={rgb, 255:red, 255; green, 255; blue, 255 }  ,draw opacity=1 ][fill={rgb, 255:red, 255; green, 255; blue, 255 }  ,fill opacity=1 ] (205,60) -- (275,60) -- (275,123) -- (205,123) -- cycle ;
\draw    (240,91.5) -- (275,128.5) ;
\draw    (240,91.5) -- (236,119.5) ;
\draw    (236,119.5) -- (275,128.5) ;
\draw    (268.5,53) -- (275,128.5) ;
\draw    (235.5,64.5) -- (268.5,53) ;
\draw    (240,91.5) -- (268.5,53) ;
\draw    (235.5,64.5) -- (240,91.5) ;
\draw  [fill={rgb, 255:red, 255; green, 255; blue, 255 }  ,fill opacity=1 ] (230,64.5) .. controls (230,61.46) and (232.46,59) .. (235.5,59) .. controls (238.54,59) and (241,61.46) .. (241,64.5) .. controls (241,67.54) and (238.54,70) .. (235.5,70) .. controls (232.46,70) and (230,67.54) .. (230,64.5) -- cycle ;
\draw  [fill={rgb, 255:red, 255; green, 255; blue, 255 }  ,fill opacity=1 ] (234.5,91.5) .. controls (234.5,88.46) and (236.96,86) .. (240,86) .. controls (243.04,86) and (245.5,88.46) .. (245.5,91.5) .. controls (245.5,94.54) and (243.04,97) .. (240,97) .. controls (236.96,97) and (234.5,94.54) .. (234.5,91.5) -- cycle ;
\draw  [fill={rgb, 255:red, 255; green, 255; blue, 255 }  ,fill opacity=1 ] (230.5,119.5) .. controls (230.5,116.46) and (232.96,114) .. (236,114) .. controls (239.04,114) and (241.5,116.46) .. (241.5,119.5) .. controls (241.5,122.54) and (239.04,125) .. (236,125) .. controls (232.96,125) and (230.5,122.54) .. (230.5,119.5) -- cycle ;
\draw  [fill={rgb, 255:red, 255; green, 255; blue, 255 }  ,fill opacity=1 ] (263,53) .. controls (263,49.96) and (265.46,47.5) .. (268.5,47.5) .. controls (271.54,47.5) and (274,49.96) .. (274,53) .. controls (274,56.04) and (271.54,58.5) .. (268.5,58.5) .. controls (265.46,58.5) and (263,56.04) .. (263,53) -- cycle ;
\draw  [fill={rgb, 255:red, 255; green, 255; blue, 255 }  ,fill opacity=1 ] (269.5,128.5) .. controls (269.5,125.46) and (271.96,123) .. (275,123) .. controls (278.04,123) and (280.5,125.46) .. (280.5,128.5) .. controls (280.5,131.54) and (278.04,134) .. (275,134) .. controls (271.96,134) and (269.5,131.54) .. (269.5,128.5) -- cycle ;
\draw  [color={rgb, 255:red, 155; green, 155; blue, 155 }  ,draw opacity=1 ] (263,34) .. controls (288,23) and (313,129) .. (295,142) .. controls (277,155) and (216,142) .. (219,119) .. controls (222,96) and (247,115) .. (260,97) .. controls (273,79) and (238,45) .. (263,34) -- cycle ;

\draw (218,82.4) node [anchor=north west][inner sep=0.75pt]    {$u$};
\draw (286,119.4) node [anchor=north west][inner sep=0.75pt]    {$v$};
\draw (301,86.4) node [anchor=north west][inner sep=0.75pt]    {$W$};

\end{tikzpicture}
&
\tikzset{every picture/.style={line width=0.75pt}} 

\begin{tikzpicture}[x=0.75pt,y=0.75pt,yscale=-1,xscale=1]

\draw  [dash pattern={on 0.84pt off 2.51pt}]  (275,128.5) -- (313,121) ;
\draw  [dash pattern={on 0.84pt off 2.51pt}]  (275,128.5) -- (313.5,142.5) ;
\draw  [dash pattern={on 0.84pt off 2.51pt}]  (275,128.5) -- (302,160) ;
\draw  [dash pattern={on 0.84pt off 2.51pt}]  (268.5,53) -- (301,41) ;
\draw  [dash pattern={on 0.84pt off 2.51pt}]  (268.5,53) -- (307,67) ;
\draw  [dash pattern={on 4.5pt off 4.5pt}] (140.7,48.1) .. controls (160.7,38.1) and (218.3,37.9) .. (230,50) .. controls (241.7,62.1) and (244.7,113.1) .. (229.7,129.1) .. controls (214.7,145.1) and (146.7,145.1) .. (135.7,129.1) .. controls (124.7,113.1) and (120.7,58.1) .. (140.7,48.1) -- cycle ;
\draw  [color={rgb, 255:red, 255; green, 255; blue, 255 }  ,draw opacity=1 ][fill={rgb, 255:red, 255; green, 255; blue, 255 }  ,fill opacity=1 ] (205,60) -- (275,60) -- (275,123) -- (205,123) -- cycle ;
\draw    (240,91.5) -- (275,128.5) ;
\draw    (240,91.5) -- (236,119.5) ;
\draw    (236,119.5) -- (275,128.5) ;
\draw    (268.5,53) -- (275,128.5) ;
\draw    (235.5,64.5) -- (268.5,53) ;
\draw    (240,91.5) -- (268.5,53) ;
\draw    (235.5,64.5) -- (240,91.5) ;
\draw  [fill={rgb, 255:red, 255; green, 255; blue, 255 }  ,fill opacity=1 ] (230,64.5) .. controls (230,61.46) and (232.46,59) .. (235.5,59) .. controls (238.54,59) and (241,61.46) .. (241,64.5) .. controls (241,67.54) and (238.54,70) .. (235.5,70) .. controls (232.46,70) and (230,67.54) .. (230,64.5) -- cycle ;
\draw  [fill={rgb, 255:red, 255; green, 255; blue, 255 }  ,fill opacity=1 ] (234.5,91.5) .. controls (234.5,88.46) and (236.96,86) .. (240,86) .. controls (243.04,86) and (245.5,88.46) .. (245.5,91.5) .. controls (245.5,94.54) and (243.04,97) .. (240,97) .. controls (236.96,97) and (234.5,94.54) .. (234.5,91.5) -- cycle ;
\draw  [fill={rgb, 255:red, 255; green, 255; blue, 255 }  ,fill opacity=1 ] (230.5,119.5) .. controls (230.5,116.46) and (232.96,114) .. (236,114) .. controls (239.04,114) and (241.5,116.46) .. (241.5,119.5) .. controls (241.5,122.54) and (239.04,125) .. (236,125) .. controls (232.96,125) and (230.5,122.54) .. (230.5,119.5) -- cycle ;
\draw  [fill={rgb, 255:red, 255; green, 255; blue, 255 }  ,fill opacity=1 ] (263,53) .. controls (263,49.96) and (265.46,47.5) .. (268.5,47.5) .. controls (271.54,47.5) and (274,49.96) .. (274,53) .. controls (274,56.04) and (271.54,58.5) .. (268.5,58.5) .. controls (265.46,58.5) and (263,56.04) .. (263,53) -- cycle ;
\draw  [fill={rgb, 255:red, 255; green, 255; blue, 255 }  ,fill opacity=1 ] (269.5,128.5) .. controls (269.5,125.46) and (271.96,123) .. (275,123) .. controls (278.04,123) and (280.5,125.46) .. (280.5,128.5) .. controls (280.5,131.54) and (278.04,134) .. (275,134) .. controls (271.96,134) and (269.5,131.54) .. (269.5,128.5) -- cycle ;
\draw  [color={rgb, 255:red, 155; green, 155; blue, 155 }  ,draw opacity=1 ] (263,34) .. controls (288,23) and (298,135) .. (287,144) .. controls (276,153) and (234,149) .. (224,133) .. controls (214,117) and (204,88) .. (211,80) .. controls (218,72) and (237.95,82.15) .. (246,71) .. controls (254.05,59.85) and (238,45) .. (263,34) -- cycle ;

\draw (218,82.4) node [anchor=north west][inner sep=0.75pt]    {$u$};
\draw (292,84.4) node [anchor=north west][inner sep=0.75pt]    {$W$};
\draw (214,50.4) node [anchor=north west][inner sep=0.75pt]    {$a$};

\end{tikzpicture}
\end{tabular}
\end{center}
\caption{\label{fig:proof-MDS} The two cases of the proof of Theorem~\ref{thm:MDS}: with an inside node, on the left, and without an inside node on the right. The cycle is represented by the dashed line, and the dotted lines represent outgoing edges of non-inside nodes.}
\end{figure}
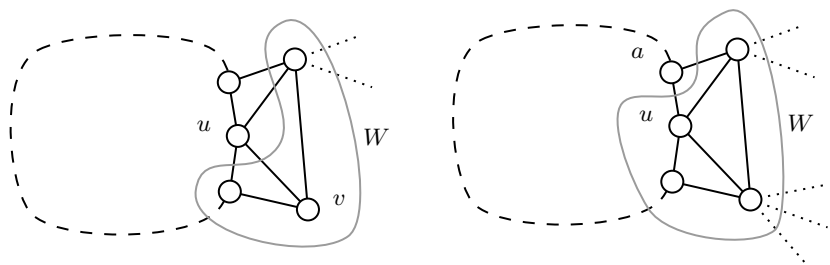


\section{Maximal matching}
\label{sec:MM}

We now turn our attention to the problem of maximal matching, and get the following theorem.

\begin{theorem}\label{thm:max-matching}
The class $\UU^1_{MM}$ is composed of the set of trees, of balanced complete bipartite graphs, and  of cliques with an even number of nodes. 
For any $k\geq 2$, the class $\UU^{k}_{MM}$ is composed of the cycle on four nodes and of the set of trees. 
\end{theorem}

The core of this part is the study of the case where only one edge is removed. At the end of the section we consider the more general technically less interesting case of multiple edges removal.

\subsection{One edge removal}

In this subsection we characterize the class of graphs where every maximal matching is 1-robust. 

\begin{lemma}\label{lem:universel-MM-1}
$\UU^1_{MM}$ is composed of the set of trees, of balanced complete bipartite graphs, and  of cliques with an even number of nodes. 
\end{lemma}

The rest of this subsection is devoted to the proof of Lemma~\ref{lem:universel-MM-1}.

\paragraph*{A result about perfect matchings}

The core of the proof is to show a connection to perfect matchings. Once this is done, we can use the following theorem from~\cite{Summer79}.

\begin{theorem}[\cite{Summer79}]\label{thm:summer79}
The class of graphs such that any maximal matching is perfect is the union of the balanced complete bipartite graphs and of the cliques of even size.
\end{theorem}

\paragraph*{First inclusion}

We start with the easy direction of the theorem, which is to prove that the graphs we mentioned are  in $\UU^1_{MM}$. 
In trees, any property is robust, since no edge can be removed without disconnecting the graph.
For the two other types, we will use the following claim.

\begin{claim}
Perfect matchings are 1-robust maximal matchings. 
\end{claim}

Consider a perfect matching in a graph, and remove an arbitrary edge (that does not disconnect the graph). 
If this edge was not in the matching, and then we still have a perfect matching, thus a maximal matching. If this edge was in the matching, then there are only two non-matched nodes in the graph (the ones that were adjacent to the edge), and all their neighbours are matched, thus the matching is still maximal. This proves the claim.\claimend{}

In balanced complete bipartite graphs and cliques of even size, any maximal matching is perfect (Theorem~\ref{thm:summer79}), and since perfect matchings are 1-robust maximal matchings, we get the first direction of Lemma~\ref{lem:universel-MM-1}.

\paragraph*{Second inclusion: three useful claims}
We now tackle the other direction. The following lemma establishes a local condition that 1-robust matchings must satisfy. See Figure~\ref{fig:u-not-matched} for an illustration.

\begin{claim}\label{clm:u-not-matched}
In a 1-robust maximal matching $M$, if a node $u$ is not matched, then all the nodes of $N(u)$ are matched, and their matched edges are bridges of the graph.
\end{claim}

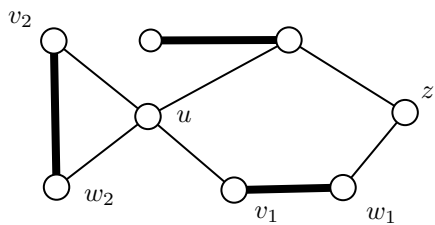
\begin{figure}[!h]
\begin{center}
\tikzset{every picture/.style={line width=0.75pt}} 

\begin{tikzpicture}[x=0.75pt,y=0.75pt,yscale=-1,xscale=1]

\draw    (194.35,183.35) -- (151.35,146.35) ;
\draw    (221.7,107.85) -- (151.35,146.35) ;
\draw [line width=3]    (248.7,182.1) -- (194.35,183.35) ;
\draw    (279.7,144.35) -- (248.7,182.1) ;
\draw    (221.7,107.85) -- (279.7,144.35) ;
\draw [line width=3]    (152.7,108.1) -- (221.7,107.85) ;
\draw    (105.7,181.87) -- (151.35,146.35) ;
\draw    (104.35,108.35) -- (151.35,146.35) ;
\draw [line width=3]    (104.35,108.35) -- (105.7,181.87) ;
\draw  [fill={rgb, 255:red, 255; green, 255; blue, 255 }  ,fill opacity=1 ] (98,108.35) .. controls (98,104.84) and (100.84,102) .. (104.35,102) .. controls (107.86,102) and (110.7,104.84) .. (110.7,108.35) .. controls (110.7,111.86) and (107.86,114.7) .. (104.35,114.7) .. controls (100.84,114.7) and (98,111.86) .. (98,108.35) -- cycle ;
\draw  [fill={rgb, 255:red, 255; green, 255; blue, 255 }  ,fill opacity=1 ] (145,146.35) .. controls (145,142.84) and (147.84,140) .. (151.35,140) .. controls (154.86,140) and (157.7,142.84) .. (157.7,146.35) .. controls (157.7,149.86) and (154.86,152.7) .. (151.35,152.7) .. controls (147.84,152.7) and (145,149.86) .. (145,146.35) -- cycle ;
\draw  [fill={rgb, 255:red, 255; green, 255; blue, 255 }  ,fill opacity=1 ] (99.35,181.87) .. controls (99.35,178.36) and (102.19,175.52) .. (105.7,175.52) .. controls (109.21,175.52) and (112.05,178.36) .. (112.05,181.87) .. controls (112.05,185.37) and (109.21,188.22) .. (105.7,188.22) .. controls (102.19,188.22) and (99.35,185.37) .. (99.35,181.87) -- cycle ;
\draw  [fill={rgb, 255:red, 255; green, 255; blue, 255 }  ,fill opacity=1 ] (215.35,107.85) .. controls (215.35,104.34) and (218.19,101.5) .. (221.7,101.5) .. controls (225.21,101.5) and (228.05,104.34) .. (228.05,107.85) .. controls (228.05,111.36) and (225.21,114.2) .. (221.7,114.2) .. controls (218.19,114.2) and (215.35,111.36) .. (215.35,107.85) -- cycle ;
\draw  [fill={rgb, 255:red, 255; green, 255; blue, 255 }  ,fill opacity=1 ] (188,183.35) .. controls (188,179.84) and (190.84,177) .. (194.35,177) .. controls (197.86,177) and (200.7,179.84) .. (200.7,183.35) .. controls (200.7,186.86) and (197.86,189.7) .. (194.35,189.7) .. controls (190.84,189.7) and (188,186.86) .. (188,183.35) -- cycle ;
\draw  [fill={rgb, 255:red, 255; green, 255; blue, 255 }  ,fill opacity=1 ] (242.35,182.1) .. controls (242.35,178.59) and (245.19,175.75) .. (248.7,175.75) .. controls (252.21,175.75) and (255.05,178.59) .. (255.05,182.1) .. controls (255.05,185.61) and (252.21,188.45) .. (248.7,188.45) .. controls (245.19,188.45) and (242.35,185.61) .. (242.35,182.1) -- cycle ;
\draw  [fill={rgb, 255:red, 255; green, 255; blue, 255 }  ,fill opacity=1 ] (273.35,144.35) .. controls (273.35,140.84) and (276.19,138) .. (279.7,138) .. controls (283.21,138) and (286.05,140.84) .. (286.05,144.35) .. controls (286.05,147.86) and (283.21,150.7) .. (279.7,150.7) .. controls (276.19,150.7) and (273.35,147.86) .. (273.35,144.35) -- cycle ;
\draw  [fill={rgb, 255:red, 255; green, 255; blue, 255 }  ,fill opacity=1 ] (147.2,108.1) .. controls (147.2,105.06) and (149.66,102.6) .. (152.7,102.6) .. controls (155.74,102.6) and (158.2,105.06) .. (158.2,108.1) .. controls (158.2,111.14) and (155.74,113.6) .. (152.7,113.6) .. controls (149.66,113.6) and (147.2,111.14) .. (147.2,108.1) -- cycle ;

\draw (164,142) node [anchor=north west][inner sep=0.75pt]    {$u$};
\draw (203,190) node [anchor=north west][inner sep=0.75pt]    {$v_{1}$};
\draw (80,92) node [anchor=north west][inner sep=0.75pt]    {$v_{2}$};
\draw (118,181.4) node [anchor=north west][inner sep=0.75pt]    {$w_{2}$};
\draw (258.7,190.85) node [anchor=north west][inner sep=0.75pt]    {$w_{1}$};
\draw (286,130) node [anchor=north west][inner sep=0.75pt]    {$z$};

\end{tikzpicture}
\end{center}
\caption{Illustration of Claim~\ref{clm:u-not-matched}. Here we have a maximal matching, and in particular all the neighbors of $u$ are matched, but it is not a 1-robust matching. 
Indeed, removing $(v_1,w_1)$ gives the possibility of adding $(u,v_1)$ and $(w_1,z)$. Also, having a triangle with a matched edge and an unmatched node, like $(u,v2,w_2)$ is impossible (Claim~\ref{clm:triangle}), since removing $(v_2,w_2)$ gives the possibility of adding either $(u,v_2)$ or $(u,w_2)$ to the matching, contradicting the maximality. Hence we need the bridge condition.}
\label{fig:u-not-matched}
\end{figure}

The fact that all the nodes in $N(u)$ are matched follows from $M$ being a maximal matching.  
Now, suppose that there exists $(v,w)\in M$, such that $v\in N(u)$ and $(v,w)$ is not a bridge. 
In other words, the removal of $(v,w)$ does not disconnect the graph. 
After this removal, both $u$ and $v$ are unmatched, and since $(u,v)$ is an edge of the graph, the matching in the new graph cannot be maximal. This contradicts the 1-robustness of $M$, and proves the lemma. 
\claimend

The following claim follows directly from Claim~\ref{clm:u-not-matched}.

\begin{claim}\label{clm:triangle}
In a 1-robust maximal matching $M$, if there is an unmatched node $u$, two nodes $a,b\in N(u)$ with $(a,b)\in E$, then $(a,b)\notin M$.
\end{claim}

We now study the shape of 1-robust maximal matchings in cycles. 

\begin{claim}\label{clm:cycle}
In every maximal matching of a graph in $\mathcal{U}^1_{MM}$, if a node belongs to a cycle, then it is matched. 
\end{claim}

Our proof of Claim~\ref{clm:cycle} consists in proving that if a maximal matching does not satisfy the condition, then either it is not 1-robust, or we can use it to build another maximal matching that is not 1-robust. In both cases this means the graph was not in $\mathcal{U}^1_{MM}$.

Consider a node $u$ in a cycle. Let $a$ and $b$ be its direct neighbors in the cycle, and let its other neighbors be $(c_i)_i$. There can be several configurations, with $a$ adjacent to $b$ or not, etc. The proof is generic to all these cases, but Figure~\ref{fig:cycle} illustrates different cases.  
Consider a 1-robust maximal matching $M$ where $u$ is unmatched.
Because of Claim~\ref{clm:u-not-matched}, we know that there exists nodes $a'$, $b'$, and $c_i'$ for all $i$, such that respectively $(a,a')$, $(b,b')$ and $(c_i,c_i')$ (for all $i$) are bridges of the graph. Because of the bridge condition, these nodes $a'$, $b'$ and $c_i'$ (for all $i$) are all different, and are different from $a$, $b$, $u$ and the $c_i$'s.
Let us also denote $d$ the neighbor of $a$ in the cycle that is not $u$. 
Note that $d$ can be a $c_i$ or $b$, but no other named node. (See Figure~\ref{fig:cycle} for an illustration.)
Now we create a new matching $M'$ from $M$ in the following way. 
First remove all the edge of the matching that are not adjacent to one of the nodes above. 
Then, remove $(a,a')$ and any edge matching $d$ (if it exists). Note that this last edge matching $d$ could be a $(c_j,c_j')$ or $(b,b')$.
Add $(a,d)$ to the matching (note that both nodes are unmatched before this operation). In this matching, all the neighbors of $u$ are matched. We complete this matching into a maximal matching $M'$. The edge $(a,d)$ is in $M'$ and $u$ is unmatched, which is a contradiction with Claim~\ref{clm:u-not-matched}, thus $M'$ cannot be 1-robust, and this proves the claim. \claimend

\begin{figure}[!h]
\begin{center}
\begin{tabular}{ccc}
\scalebox{0.8}{
\tikzset{every picture/.style={line width=0.75pt}} 

\begin{tikzpicture}[x=0.75pt,y=0.75pt,yscale=-1,xscale=1]

\draw [line width=3]    (283.35,126.35) -- (234.35,127.35) ;
\draw [line width=3]    (282.35,84.35) -- (233.35,85.35) ;
\draw    (233.35,85.35) -- (184.35,86.35) ;
\draw    (234.35,127.35) -- (184.35,86.35) ;
\draw [line width=3]    (247.35,173.35) -- (198.35,174.35) ;
\draw  [dash pattern={on 4.5pt off 4.5pt}]  (199.35,174.35) .. controls (222.7,169.1) and (234.7,147.1) .. (234.35,127.35) ;
\draw  [dash pattern={on 4.5pt off 4.5pt}]  (134.35,173.35) .. controls (152.7,186.1) and (180.7,193.1) .. (199.35,174.35) ;
\draw    (184.35,86.35) -- (199.35,174.35) ;
\draw    (134.35,173.35) -- (133.35,128.35) ;
\draw    (184.35,86.35) -- (133.35,128.35) ;
\draw [line width=2.25]    (133.35,128.35) -- (84.35,128.35) ;
\draw  [fill={rgb, 255:red, 255; green, 255; blue, 255 }  ,fill opacity=1 ] (178,86.35) .. controls (178,82.84) and (180.84,80) .. (184.35,80) .. controls (187.86,80) and (190.7,82.84) .. (190.7,86.35) .. controls (190.7,89.86) and (187.86,92.7) .. (184.35,92.7) .. controls (180.84,92.7) and (178,89.86) .. (178,86.35) -- cycle ;
\draw  [fill={rgb, 255:red, 255; green, 255; blue, 255 }  ,fill opacity=1 ] (127,128.35) .. controls (127,124.84) and (129.84,122) .. (133.35,122) .. controls (136.86,122) and (139.7,124.84) .. (139.7,128.35) .. controls (139.7,131.86) and (136.86,134.7) .. (133.35,134.7) .. controls (129.84,134.7) and (127,131.86) .. (127,128.35) -- cycle ;
\draw  [fill={rgb, 255:red, 255; green, 255; blue, 255 }  ,fill opacity=1 ] (228,127.35) .. controls (228,123.84) and (230.84,121) .. (234.35,121) .. controls (237.86,121) and (240.7,123.84) .. (240.7,127.35) .. controls (240.7,130.86) and (237.86,133.7) .. (234.35,133.7) .. controls (230.84,133.7) and (228,130.86) .. (228,127.35) -- cycle ;
\draw  [fill={rgb, 255:red, 255; green, 255; blue, 255 }  ,fill opacity=1 ] (128,173.35) .. controls (128,169.84) and (130.84,167) .. (134.35,167) .. controls (137.86,167) and (140.7,169.84) .. (140.7,173.35) .. controls (140.7,176.86) and (137.86,179.7) .. (134.35,179.7) .. controls (130.84,179.7) and (128,176.86) .. (128,173.35) -- cycle ;
\draw  [fill={rgb, 255:red, 255; green, 255; blue, 255 }  ,fill opacity=1 ] (193,174.35) .. controls (193,170.84) and (195.84,168) .. (199.35,168) .. controls (202.86,168) and (205.7,170.84) .. (205.7,174.35) .. controls (205.7,177.86) and (202.86,180.7) .. (199.35,180.7) .. controls (195.84,180.7) and (193,177.86) .. (193,174.35) -- cycle ;
\draw  [fill={rgb, 255:red, 255; green, 255; blue, 255 }  ,fill opacity=1 ] (227,85.35) .. controls (227,81.84) and (229.84,79) .. (233.35,79) .. controls (236.86,79) and (239.7,81.84) .. (239.7,85.35) .. controls (239.7,88.86) and (236.86,91.7) .. (233.35,91.7) .. controls (229.84,91.7) and (227,88.86) .. (227,85.35) -- cycle ;
\draw  [fill={rgb, 255:red, 255; green, 255; blue, 255 }  ,fill opacity=1 ] (277,126.35) .. controls (277,122.84) and (279.84,120) .. (283.35,120) .. controls (286.86,120) and (289.7,122.84) .. (289.7,126.35) .. controls (289.7,129.86) and (286.86,132.7) .. (283.35,132.7) .. controls (279.84,132.7) and (277,129.86) .. (277,126.35) -- cycle ;
\draw  [fill={rgb, 255:red, 255; green, 255; blue, 255 }  ,fill opacity=1 ] (241,173.35) .. controls (241,169.84) and (243.84,167) .. (247.35,167) .. controls (250.86,167) and (253.7,169.84) .. (253.7,173.35) .. controls (253.7,176.86) and (250.86,179.7) .. (247.35,179.7) .. controls (243.84,179.7) and (241,176.86) .. (241,173.35) -- cycle ;
\draw  [fill={rgb, 255:red, 255; green, 255; blue, 255 }  ,fill opacity=1 ] (78,128.35) .. controls (78,124.84) and (80.84,122) .. (84.35,122) .. controls (87.86,122) and (90.7,124.84) .. (90.7,128.35) .. controls (90.7,131.86) and (87.86,134.7) .. (84.35,134.7) .. controls (80.84,134.7) and (78,131.86) .. (78,128.35) -- cycle ;
\draw  [fill={rgb, 255:red, 255; green, 255; blue, 255 }  ,fill opacity=1 ] (276,84.35) .. controls (276,80.84) and (278.84,78) .. (282.35,78) .. controls (285.86,78) and (288.7,80.84) .. (288.7,84.35) .. controls (288.7,87.86) and (285.86,90.7) .. (282.35,90.7) .. controls (278.84,90.7) and (276,87.86) .. (276,84.35) -- cycle ;
\draw [line width=1.5]  [dash pattern={on 1.69pt off 2.76pt}]  (117.7,117.1) -- (98.7,141.1) ;
\draw [line width=1.5]  [dash pattern={on 1.69pt off 2.76pt}]  (141.7,137.1) -- (142.7,164.1) ;

\draw (165,72) node [anchor=north west][inner sep=0.75pt]    {$u$};
\draw (144,122) node [anchor=north west][inner sep=0.75pt]    {$a$};
\draw (144,165) node [anchor=north west][inner sep=0.75pt]    {$d$};
\draw (213,119) node [anchor=north west][inner sep=0.75pt]    {$b$};
\draw (200.35,183) node [anchor=north west][inner sep=0.75pt]    {$c_{1}$};
\draw (248.35,181.1) node [anchor=north west][inner sep=0.75pt]    {$c_{1} '$};
\draw (227.35,65) node [anchor=north west][inner sep=0.75pt]    {$c_{2}$};
\draw (276.35,60) node [anchor=north west][inner sep=0.75pt]    {$
c_{2}'$};
\draw (275,136.4) node [anchor=north west][inner sep=0.75pt]    {$b'$};
\draw (79,105) node [anchor=north west][inner sep=0.75pt]    {$a'$};

\end{tikzpicture}}
&
\scalebox{0.8}{
\tikzset{every picture/.style={line width=0.75pt}} 

\begin{tikzpicture}[x=0.75pt,y=0.75pt,yscale=-1,xscale=1]

\draw [line width=3]    (283.35,126.35) -- (234.35,127.35) ;
\draw [line width=3]    (282.35,84.35) -- (233.35,85.35) ;
\draw    (233.35,85.35) -- (184.35,86.35) ;
\draw    (234.35,127.35) -- (184.35,86.35) ;
\draw [line width=3]    (201.35,178.35) -- (152.35,179.35) ;
\draw  [dash pattern={on 4.5pt off 4.5pt}]  (152.35,179.35) .. controls (180.7,163.1) and (214.7,143.1) .. (234.35,127.35) ;
\draw    (184.35,86.35) -- (152.35,179.35) ;
\draw    (152.35,179.35) -- (133.35,128.35) ;
\draw    (184.35,86.35) -- (133.35,128.35) ;
\draw [line width=2.25]    (133.35,128.35) -- (84.35,128.35) ;
\draw  [fill={rgb, 255:red, 255; green, 255; blue, 255 }  ,fill opacity=1 ] (178,86.35) .. controls (178,82.84) and (180.84,80) .. (184.35,80) .. controls (187.86,80) and (190.7,82.84) .. (190.7,86.35) .. controls (190.7,89.86) and (187.86,92.7) .. (184.35,92.7) .. controls (180.84,92.7) and (178,89.86) .. (178,86.35) -- cycle ;
\draw  [fill={rgb, 255:red, 255; green, 255; blue, 255 }  ,fill opacity=1 ] (127,128.35) .. controls (127,124.84) and (129.84,122) .. (133.35,122) .. controls (136.86,122) and (139.7,124.84) .. (139.7,128.35) .. controls (139.7,131.86) and (136.86,134.7) .. (133.35,134.7) .. controls (129.84,134.7) and (127,131.86) .. (127,128.35) -- cycle ;
\draw  [fill={rgb, 255:red, 255; green, 255; blue, 255 }  ,fill opacity=1 ] (228,127.35) .. controls (228,123.84) and (230.84,121) .. (234.35,121) .. controls (237.86,121) and (240.7,123.84) .. (240.7,127.35) .. controls (240.7,130.86) and (237.86,133.7) .. (234.35,133.7) .. controls (230.84,133.7) and (228,130.86) .. (228,127.35) -- cycle ;
\draw  [fill={rgb, 255:red, 255; green, 255; blue, 255 }  ,fill opacity=1 ] (146,179.35) .. controls (146,175.84) and (148.84,173) .. (152.35,173) .. controls (155.86,173) and (158.7,175.84) .. (158.7,179.35) .. controls (158.7,182.86) and (155.86,185.7) .. (152.35,185.7) .. controls (148.84,185.7) and (146,182.86) .. (146,179.35) -- cycle ;
\draw  [fill={rgb, 255:red, 255; green, 255; blue, 255 }  ,fill opacity=1 ] (227,85.35) .. controls (227,81.84) and (229.84,79) .. (233.35,79) .. controls (236.86,79) and (239.7,81.84) .. (239.7,85.35) .. controls (239.7,88.86) and (236.86,91.7) .. (233.35,91.7) .. controls (229.84,91.7) and (227,88.86) .. (227,85.35) -- cycle ;
\draw  [fill={rgb, 255:red, 255; green, 255; blue, 255 }  ,fill opacity=1 ] (277,126.35) .. controls (277,122.84) and (279.84,120) .. (283.35,120) .. controls (286.86,120) and (289.7,122.84) .. (289.7,126.35) .. controls (289.7,129.86) and (286.86,132.7) .. (283.35,132.7) .. controls (279.84,132.7) and (277,129.86) .. (277,126.35) -- cycle ;
\draw  [fill={rgb, 255:red, 255; green, 255; blue, 255 }  ,fill opacity=1 ] (201.35,178.35) .. controls (201.35,174.84) and (204.19,172) .. (207.7,172) .. controls (211.21,172) and (214.05,174.84) .. (214.05,178.35) .. controls (214.05,181.86) and (211.21,184.7) .. (207.7,184.7) .. controls (204.19,184.7) and (201.35,181.86) .. (201.35,178.35) -- cycle ;
\draw  [fill={rgb, 255:red, 255; green, 255; blue, 255 }  ,fill opacity=1 ] (78,128.35) .. controls (78,124.84) and (80.84,122) .. (84.35,122) .. controls (87.86,122) and (90.7,124.84) .. (90.7,128.35) .. controls (90.7,131.86) and (87.86,134.7) .. (84.35,134.7) .. controls (80.84,134.7) and (78,131.86) .. (78,128.35) -- cycle ;
\draw  [fill={rgb, 255:red, 255; green, 255; blue, 255 }  ,fill opacity=1 ] (276,84.35) .. controls (276,80.84) and (278.84,78) .. (282.35,78) .. controls (285.86,78) and (288.7,80.84) .. (288.7,84.35) .. controls (288.7,87.86) and (285.86,90.7) .. (282.35,90.7) .. controls (278.84,90.7) and (276,87.86) .. (276,84.35) -- cycle ;
\draw [line width=1.5]  [dash pattern={on 1.69pt off 2.76pt}]  (117.7,117.1) -- (98.7,141.1) ;
\draw [line width=1.5]  [dash pattern={on 1.69pt off 2.76pt}]  (141.7,137.1) -- (151.7,165.1) ;
\draw [line width=1.5]  [dash pattern={on 1.69pt off 2.76pt}]  (194.7,167.1) -- (175.7,191.1) ;

\draw (162,69.4) node [anchor=north west][inner sep=0.75pt]    {$u$};
\draw (146,120.4) node [anchor=north west][inner sep=0.75pt]    {$a$};
\draw (211,117.4) node [anchor=north west][inner sep=0.75pt]    {$b$};
\draw (154.35,182.75) node [anchor=north west][inner sep=0.75pt]    {$c_{1}$};
\draw (209.7,181.75) node [anchor=north west][inner sep=0.75pt]    {$c_{1} '$};
\draw (227.35,55.1) node [anchor=north west][inner sep=0.75pt]    {$c_{2}$};
\draw (276.35,55.1) node [anchor=north west][inner sep=0.75pt]    {$ \begin{array}{l}
c_{2} '\\
\end{array}$};
\draw (275,136.4) node [anchor=north west][inner sep=0.75pt]    {$b'$};
\draw (79,100.4) node [anchor=north west][inner sep=0.75pt]    {$a'$};

\end{tikzpicture}}
&
\scalebox{0.8}{
\tikzset{every picture/.style={line width=0.75pt}} 

\begin{tikzpicture}[x=0.75pt,y=0.75pt,yscale=-1,xscale=1]

\draw [line width=3]    (212.35,174.35) -- (234.35,127.35) ;
\draw [line width=3]    (201.35,46.35) -- (152.35,47.35) ;
\draw    (152.35,47.35) -- (184.35,86.35) ;
\draw    (184.35,86.35) -- (234.35,127.35) ;
\draw    (234.35,127.35) -- (133.35,128.35) ;
\draw    (184.35,86.35) -- (133.35,128.35) ;
\draw [line width=2.25]    (133.35,128.35) -- (152.35,170.35) ;
\draw  [fill={rgb, 255:red, 255; green, 255; blue, 255 }  ,fill opacity=1 ] (178,86.35) .. controls (178,82.84) and (180.84,80) .. (184.35,80) .. controls (187.86,80) and (190.7,82.84) .. (190.7,86.35) .. controls (190.7,89.86) and (187.86,92.7) .. (184.35,92.7) .. controls (180.84,92.7) and (178,89.86) .. (178,86.35) -- cycle ;
\draw  [fill={rgb, 255:red, 255; green, 255; blue, 255 }  ,fill opacity=1 ] (127,128.35) .. controls (127,124.84) and (129.84,122) .. (133.35,122) .. controls (136.86,122) and (139.7,124.84) .. (139.7,128.35) .. controls (139.7,131.86) and (136.86,134.7) .. (133.35,134.7) .. controls (129.84,134.7) and (127,131.86) .. (127,128.35) -- cycle ;
\draw  [fill={rgb, 255:red, 255; green, 255; blue, 255 }  ,fill opacity=1 ] (228,127.35) .. controls (228,123.84) and (230.84,121) .. (234.35,121) .. controls (237.86,121) and (240.7,123.84) .. (240.7,127.35) .. controls (240.7,130.86) and (237.86,133.7) .. (234.35,133.7) .. controls (230.84,133.7) and (228,130.86) .. (228,127.35) -- cycle ;
\draw  [fill={rgb, 255:red, 255; green, 255; blue, 255 }  ,fill opacity=1 ] (206,174.35) .. controls (206,170.84) and (208.84,168) .. (212.35,168) .. controls (215.86,168) and (218.7,170.84) .. (218.7,174.35) .. controls (218.7,177.86) and (215.86,180.7) .. (212.35,180.7) .. controls (208.84,180.7) and (206,177.86) .. (206,174.35) -- cycle ;
\draw  [fill={rgb, 255:red, 255; green, 255; blue, 255 }  ,fill opacity=1 ] (146,47.35) .. controls (146,43.84) and (148.84,41) .. (152.35,41) .. controls (155.86,41) and (158.7,43.84) .. (158.7,47.35) .. controls (158.7,50.86) and (155.86,53.7) .. (152.35,53.7) .. controls (148.84,53.7) and (146,50.86) .. (146,47.35) -- cycle ;
\draw  [fill={rgb, 255:red, 255; green, 255; blue, 255 }  ,fill opacity=1 ] (146,170.35) .. controls (146,166.84) and (148.84,164) .. (152.35,164) .. controls (155.86,164) and (158.7,166.84) .. (158.7,170.35) .. controls (158.7,173.86) and (155.86,176.7) .. (152.35,176.7) .. controls (148.84,176.7) and (146,173.86) .. (146,170.35) -- cycle ;
\draw  [fill={rgb, 255:red, 255; green, 255; blue, 255 }  ,fill opacity=1 ] (201.35,46.35) .. controls (201.35,42.84) and (204.19,40) .. (207.7,40) .. controls (211.21,40) and (214.05,42.84) .. (214.05,46.35) .. controls (214.05,49.86) and (211.21,52.7) .. (207.7,52.7) .. controls (204.19,52.7) and (201.35,49.86) .. (201.35,46.35) -- cycle ;
\draw [line width=1.5]  [dash pattern={on 1.69pt off 2.76pt}]  (157.7,143.1) -- (127.7,157.1) ;
\draw [line width=1.5]  [dash pattern={on 1.69pt off 2.76pt}]  (150.7,120.1) -- (215.7,119.1) ;
\draw [line width=1.5]  [dash pattern={on 1.69pt off 2.76pt}]  (237.35,155.85) -- (209.35,145.85) ;

\draw (157,78.4) node [anchor=north west][inner sep=0.75pt]    {$u$};
\draw (129,101.4) node [anchor=north west][inner sep=0.75pt]    {$a$};
\draw (229,100.4) node [anchor=north west][inner sep=0.75pt]    {$b$};
\draw (133.35,46.1) node [anchor=north west][inner sep=0.75pt]    {$c_{2}$};
\draw (209.7,49.75) node [anchor=north west][inner sep=0.75pt]    {$c_{2} '$};
\draw (225,166.4) node [anchor=north west][inner sep=0.75pt]    {$b'$};
\draw (163,165.4) node [anchor=north west][inner sep=0.75pt]    {$a'$};

\end{tikzpicture}}
\end{tabular}
\end{center}
\caption{\label{fig:cycle}Illustration of the proof of Claim~\ref{clm:no-hybrid}, in three cases: $d$ is not $b$ nor a $c_i$, $d$ is one of the $c_i$, $d$ is $b$. The dashed lines represent paths with at least one edge. The dotted lines represent the change we operate: the edges that are crossed out are removed from the matching, the edges that are have a dotted double are added to the matching.}
\end{figure}
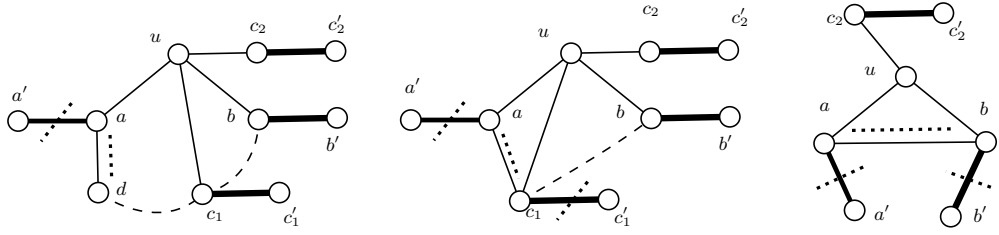

\paragraph*{Second inclusion: putting pieces together}

\begin{claim}\label{clm:no-hybrid}
A graph in the class $\UU^1_{MM}$ is either a tree or is 2-connected.
\end{claim}

Consider a graph that is neither a tree nor a 2-connected graph. 
There necessarily exists a bridge $(u,v)$ such that $u$ belongs to a cycle. 
We distinguish two cases. 
\begin{enumerate}
\item Node $v$ is linked only to $u$, that is, $v$ is a pendant node. Then we build a maximal matching $M$ by first forcing $u$ to be matched to a node that is not $v$, and then completing it greedily. Now, if we remove the edge that matches $u$, we do not disconnect the graph since $u$ was part of a cycle, but neither $u$ nor $v$ is matched, thus the matching is not maximal ($(u,v)$ could be added).
Thus the matching $M$ was not 1-robustness.

\item Node $v$ is linked to another node $w$. Let consider $(v_i)_i$ the set of nodes such that $v_i\ne v$ and $(u,v_i)$ is a bridge. By the previous point, we know that there exists some $w_i\ne u$ in $N(v_i)$. Moreover, those $(w_i)$ must be distinct pairwise and from all the other named nodes, otherwise $(u,v_i)$ would not have been a bridge.
The node $w$ and the nodes $(w_i)_i$ cannot be part of the 2-connected component of $u$, otherwise $(u,v)$  and $(u,v_i)$ 
would not be a bridge. 
We build a maximal matching~$M$ by first forcing $(u,v)$  and $(v_i,w_i)$ for all i, and then completing it greedily.
As observed earlier, in the 2-connected component of $u$ every node must belong to a cycle, thus 
by Claim~\ref{clm:cycle}, we get that every node of this component must be matched. 
We now build a second matching $M'$. We start from $M$ and remove from the matching $(u,v)$ and every edge that is in $v$'s side of the bridge. 
Then we force $(v,w)$ in the matching, and complete it greedily. The matching $M'$ is maximal and $u$ is unmatched, since all of its neighbors are matched, hence by  Claim~\ref{clm:cycle} it is not 1-robust, since it belongs to a 2-connected components thus to a cycle. 
\end{enumerate}
This concludes the proof of the claim.
\claimend

To conclude a graph in the class is either a tree, or is 2-connected, and in this last case because of Claim~\ref{clm:cycle}, every node must be matched in every maximal matching. Then Lemma~\ref{lem:universel-MM-1} follows from Theorem~\ref{thm:summer79}.

\subsection{More than one edge removal}

\begin{lemma}\label{lem:max-matching-k-more-than-2}
For any $k\geq 2$, $\UU^{k}_{MM}$ is composed of the cycle on four nodes and of the set of trees. 
\end{lemma}

\begin{proof}
We first prove the reverse inclusion. 
As before, trees are in $\UU^k_{MM}$ for any $k$ because any edge removal disconnects the graph. Then for $C_4$, note that it belongs to $\UU^1_{MM}$, and that the removal of more than one edge disconnects the graph.

For the other direction, we can restrict to $\UU^2_{MM}$, and by definition it is included in $\UU^1_{MM}$. Thus we can simply study the case of the balanced complete bipartite graphs and of the cliques on an even number of nodes. 
Consider  first a complete bipartite graph $B_{k,k}$ with $k>2$ (that is any $B_{k,k}$ larger than $C_4$), and a maximal matching $M$. 
Take two arbitrary edges $(a_1,b_1)$ and $(a_2,b_2)$ from the matching and remove them from the graph. 
The graph is still connected. 
Now the nodes $a_1$ and $b_2$ are unmatched and there is an edge between them, thus the resulting matching is not maximal and $M$ is not 2-robust. 
Thus the only $B_{k,k}$ left in the class $\UU^2_{MM}$ is $C_4$. 
For the cliques on an even number of nodes, consider one that has strictly more than two vertices. A maximal matching $M$ contains at least two edges $(u_1,v_1)$ and $(u_2,v_2)$. When we remove these edges from the graph, we still have a connected graph, $u_1$ and $u_2$ are unmatched, but $(u_1,u_2)$ still exists, thus the resulting matching is not maximal and $M$ was not 2-robust. 
\end{proof}


\section{Maximal independent set}
\label{sec:MIS}

Maximal independent set illustrates yet another behavior for the classes $(\UU_{MIS}^k)_k$: they form an infinite strict hierarchy.  

\subsection{An infinite hierarchy}

\begin{theorem}\label{thm:MIS}
For every $k\geq 1$, $\UU^{k+1}_{MIS}$ is strictly included in $\UU^{k}_{MIS}$.
\end{theorem}

\begin{proof}
Let $k\geq 1$.
We will define a graph $G_k$, and prove that it belongs to $\mathcal{U}^{k}_{MIS}$ but not to $\mathcal{U}^{k+1}_{MIS}$.

To build $G_k$, consider a bipartite graph with $k+2$ nodes on each of the sides $A$ and $B$, and add a pendant neighbor $v$ to a node $u$ on the side $A$. See Figure~\ref{fig:MIS}.
This graph has only three MIS: $A$, $v \cup B$, and $v\cup (A\setminus u)$. 
Indeed: (1) if the MIS contains $u$, then it cannot contain vertices outside of $A$, and to be maximal it contains all of $A$, (2) if it contains a vertex of $B$, it cannot contain a vertex of $A$, and by maximality it contains all of $B$ and $v$, and (3) if it contains $v$, and no vertex of $B$, then by maximality it is $v \cup (A \setminus u) $.

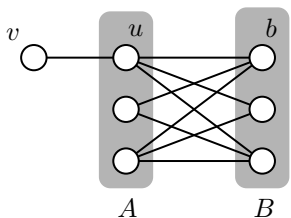
\begin{figure}[!h]
\centering
\tikzset{every picture/.style={line width=0.75pt}} 

\begin{tikzpicture}[x=0.75pt,y=0.75pt,yscale=-1,xscale=1]

\draw  [color={rgb, 255:red, 180; green, 180; blue, 180 }  ,draw opacity=1 ][fill={rgb, 255:red, 180; green, 180; blue, 180 }  ,fill opacity=1 ] (253.3,83.79) .. controls (253.3,80.9) and (255.64,78.57) .. (258.52,78.57) -- (274.18,78.57) .. controls (277.06,78.57) and (279.4,80.9) .. (279.4,83.79) -- (279.4,163.35) .. controls (279.4,166.23) and (277.06,168.57) .. (274.18,168.57) -- (258.52,168.57) .. controls (255.64,168.57) and (253.3,166.23) .. (253.3,163.35) -- cycle ;
\draw  [color={rgb, 255:red, 180; green, 180; blue, 180 }  ,draw opacity=1 ][fill={rgb, 255:red, 180; green, 180; blue, 180 }  ,fill opacity=1 ] (185.3,85.79) .. controls (185.3,82.9) and (187.64,80.57) .. (190.52,80.57) -- (206.18,80.57) .. controls (209.06,80.57) and (211.4,82.9) .. (211.4,85.79) -- (211.4,163.35) .. controls (211.4,166.23) and (209.06,168.57) .. (206.18,168.57) -- (190.52,168.57) .. controls (187.64,168.57) and (185.3,166.23) .. (185.3,163.35) -- cycle ;
\draw    (152.35,103.35) -- (266.35,103.35) ;
\draw    (198.35,103.35) -- (266.35,129.35) ;
\draw    (198.35,103.35) -- (266.35,155.35) ;
\draw    (198.35,129.35) -- (266.35,103.35) ;
\draw    (198.35,155.35) -- (266.35,129.35) ;
\draw    (198.35,129.35) -- (266.35,155.35) ;
\draw    (198.35,155.35) -- (266.35,155.35) ;
\draw    (198.35,155.35) -- (266.35,103.35) ;
\draw  [fill={rgb, 255:red, 255; green, 255; blue, 255 }  ,fill opacity=1 ] (192,103.35) .. controls (192,99.84) and (194.84,97) .. (198.35,97) .. controls (201.86,97) and (204.7,99.84) .. (204.7,103.35) .. controls (204.7,106.86) and (201.86,109.7) .. (198.35,109.7) .. controls (194.84,109.7) and (192,106.86) .. (192,103.35) -- cycle ;
\draw  [fill={rgb, 255:red, 255; green, 255; blue, 255 }  ,fill opacity=1 ] (192,129.35) .. controls (192,125.84) and (194.84,123) .. (198.35,123) .. controls (201.86,123) and (204.7,125.84) .. (204.7,129.35) .. controls (204.7,132.86) and (201.86,135.7) .. (198.35,135.7) .. controls (194.84,135.7) and (192,132.86) .. (192,129.35) -- cycle ;
\draw  [fill={rgb, 255:red, 255; green, 255; blue, 255 }  ,fill opacity=1 ] (192,155.35) .. controls (192,151.84) and (194.84,149) .. (198.35,149) .. controls (201.86,149) and (204.7,151.84) .. (204.7,155.35) .. controls (204.7,158.86) and (201.86,161.7) .. (198.35,161.7) .. controls (194.84,161.7) and (192,158.86) .. (192,155.35) -- cycle ;
\draw  [fill={rgb, 255:red, 255; green, 255; blue, 255 }  ,fill opacity=1 ] (260,103.35) .. controls (260,99.84) and (262.84,97) .. (266.35,97) .. controls (269.86,97) and (272.7,99.84) .. (272.7,103.35) .. controls (272.7,106.86) and (269.86,109.7) .. (266.35,109.7) .. controls (262.84,109.7) and (260,106.86) .. (260,103.35) -- cycle ;
\draw  [fill={rgb, 255:red, 255; green, 255; blue, 255 }  ,fill opacity=1 ] (260,129.35) .. controls (260,125.84) and (262.84,123) .. (266.35,123) .. controls (269.86,123) and (272.7,125.84) .. (272.7,129.35) .. controls (272.7,132.86) and (269.86,135.7) .. (266.35,135.7) .. controls (262.84,135.7) and (260,132.86) .. (260,129.35) -- cycle ;
\draw  [fill={rgb, 255:red, 255; green, 255; blue, 255 }  ,fill opacity=1 ] (260,155.35) .. controls (260,151.84) and (262.84,149) .. (266.35,149) .. controls (269.86,149) and (272.7,151.84) .. (272.7,155.35) .. controls (272.7,158.86) and (269.86,161.7) .. (266.35,161.7) .. controls (262.84,161.7) and (260,158.86) .. (260,155.35) -- cycle ;
\draw  [fill={rgb, 255:red, 255; green, 255; blue, 255 }  ,fill opacity=1 ] (146,103.35) .. controls (146,99.84) and (148.84,97) .. (152.35,97) .. controls (155.86,97) and (158.7,99.84) .. (158.7,103.35) .. controls (158.7,106.86) and (155.86,109.7) .. (152.35,109.7) .. controls (148.84,109.7) and (146,106.86) .. (146,103.35) -- cycle ;

\draw (192.52,173) node [anchor=north west][inner sep=0.75pt]    {$A$};
\draw (260.52,173) node [anchor=north west][inner sep=0.75pt]    {$B$};
\draw (137,87) node [anchor=north west][inner sep=0.75pt]    {$v$};
\draw (198,85) node [anchor=north west][inner sep=0.75pt]    {$u$};
\draw (266.52,81.97) node [anchor=north west][inner sep=0.75pt]    {$b$};

\end{tikzpicture}
\caption{\label{fig:MIS}. Illustration of the graph $G_k$ in the proof of Theorem~\ref{thm:MIS}.}
\end{figure}

We claim that these three MIS are $k$-robust, therefore $G_k$ is in $\mathcal{U}^{k}_{MIS}$. 
Suppose an MIS is not $k$-robust. Then there exists a vertex $w$ that is not part of the MIS, such that after at most $k$ edge removals, it has no neighbor in the MIS anymore. 
Let us make a quick case analysis depending on who is this vertex $w$. 
It cannot be $v$, since removing the edge $(u,v)$ would disconnect the graph. 
It cannot be a vertex of $A$, nor of $B$, because in all MIS mentioned, all non selected nodes (except $v$) have at least $k+1$ selected neighbors.

Now we claim that $v \cup (A \setminus u) $ is not $(k+1)$-robust, thus $G_k$ does not belong to $\mathcal{U}^{k+1}_{MIS}$.
We choose a vertex $b$ on the $B$ side, and remove all the edges $(a,b)$ for $a\in A\setminus u$. 
This is a set of $k+1$ edges whose removal does not disconnect the graph, but leaves $b$ without selected neighbors. This $v \cup (A \setminus u)$ is not $(k+1)$-robust.
\end{proof}

\subsection{A structure theorem for $\UU_{MIS}^k$}

The construction used in the proof of Theorem~\ref{thm:MIS} is very specific and does not really inform about the nature of the graphs in $\UU_{MIS}^k$.
It can be generalized, with antennas on both sides and arbitrarily large (unbalanced) bipartite graphs with arbitrary number of antennas per nodes, but it is still specific. 
Moreover these construction heavily rely on pendant nodes, that are in some sense abusing the fact that we do not worry about the correctness of the solution if the graph gets disconnected.

In order to better understand these classes, and to give a more flexible way to build such graphs, we prove a theorem about how the class behaves with respect to the join operation (Definition~\ref{def:join}).
 
We denote by $\mathcal{G}_p$ the class of graphs where every maximal independent set has size \emph{at least} $p$. 
We say that a graph class is \emph{stable by an operation} if, by applying this operation to any (set of) graph(s) from the class, the resulting graph is also in the class.

\begin{theorem}
For all $k$, the class $\UU_{MIS}^k \cap \mathcal{G}_{k+1}$ is stable by join operation. Also, if either $G$ or $H$ is not in $\UU^{k+1}_{MIS}$, then $join(G,H)$ is not in $\UU^{k+1}_{MIS}$ either.
\end{theorem}

\begin{proof}
Let us start with the first statement of the theorem.
Consider two graphs $G$ and $H$ in  $\UU_{MIS}^k \cap \mathcal{G}_{k+1}$. 
We prove that $J=join(G,H)$ is also in $\UU_{MIS}^k \cap \mathcal{G}_{k+1}$. 

\begin{claim}
Any MIS of $J$ is either completely contained in the vertex set of $G$, and is an MIS of $G$, or   contained in the vertex set of $H$, and is an MIS of $H$.  
\end{claim}

Consider an independent set in $J$. 
If it has a node $u$ in $G$, then it has no node in $H$, as by construction, all nodes of $H$ are linked to $u$. 
The analogue holds if the independent set has a node in $H$.
Thus any independent set is either completely contained in $G$ or completely contained in $H$.
Now, a set is maximal independent in $G$ (resp. $H$) alone if and only if it is maximal independent in $G$ (resp. $H$) inside $J$. Indeed the only edges that we have added are between nodes of $G$ and nodes of $H$. This proves the claim.
\claimend

Therefore, the resulting graph is in $\mathcal{G}_{k+1}$. Now for the $k$-robustness, consider without loss of generality an MIS of $J$ that is in part~$G$, and suppose it is not $k$-robust. 
In this case there must exists a non-selected vertex $v$, that has no more selected neighbors after the removal of $k$ edges (while the graph stays connected). 
This node cannot be in the part $G$, otherwise the same independent set in the graph $G$ would not be $k$-robust. 
And it cannot be in the part $H$, since every node of $H$ is linked to all the vertices of the MIS, and this set has size at least $k+1$ since $G\in \mathcal{G}_{k+1}$. 

Now, let us move on to the second statement of the theorem. Let's assume that $G$ has an MIS $S$ and $k+1$ edges such that their removal makes that $S$ is not longer maximal (i.e. there exists some $u$ that can be added to the set). Then, $S$ is also an MIS of $join(G,H)$, and the removal of the same edges will allow to add $u$ to the set, as the only new neighbors of $u$ are from $H$ that does not contain any node of the chosen MIS
\end{proof}


\section{The existence of a robust MIS is NP-hard}
\label{sec:NPH}

Remember that we have defined two types of graph classes related to robustness. For a given problem, and a parameter $k$, the universal class is the class where every solution is $k$-robust. This is the version we have explored so far. For this version, recognizing the graphs of the class is easy since these have simple explicit characterization. 
The second type of class is the existential type, where we want that there exists a solution that is $k$-robust. And here the landscape is much more complex. 
Indeed, in \cite{CasteigtsDPR20} in the simpler case of robustness without parameter, there is no explicit characterization of the existential class, only a rather involved algorithm.
In this section we show that, when we add the parameter $k$ the situation becomes even more challenging: the algorithm of \cite{CasteigtsDPR20} runs in polynomial time, and here we show that the recognition of $\mathcal{E}_{MIS}^1$ is NP-hard.

\begin{theorem}\label{thm:MIS-NP-hard}
For every odd integer $k$, it is NP-hard to decide whether a graph belongs to $\mathcal{E}_{MIS}^k$.
\end{theorem}

The rest of this section is devoted to the proof of this theorem. It is based on the NP-completeness of the following problem.\\ 

\noindent\textsc{Perfect stable}\\  
Input: A graph $G=(V,E)$.\\  
Question: Does there exists a subset of vertices $S \subset V$ that is independent 2-dominating?\\ 

Remember that a set is independent 2-dominating if no two neighbors can be selected, and every non-selected vertex should have at least two selected neighbors. Just to get some intuition about why we are interested in this problem, note that with independent 2-dominating after removing an edge between a selected and a non-selected vertex, the non-selected vertex is still dominated.
It was proved in~\cite{CroitoruS83} that \textsc{Perfect stable} is NP-hard in general. We will need the following strengthening of this hardness result. 

\begin{lemma}\label{lem:NP-perfect-stable}
Deciding whether a 2-connected graph has an independent 2-dominating set is NP-complete.
\end{lemma}

Note that this lemma does not follow directly from~\cite{CroitoruS83} because the reduction there does use some non-2-connected graphs. 

\begin{proof}
Let $G$ be an arbitrary connected graph with at least one edge. 
Consider $G'$ to be the same as $G$ but with a universal vertex, that is, $G$ with an additional vertex $u$ that is adjacent to all the vertices of $G$.  
This graph is 2-edge connected. 
Indeed, since $G$ is connected and has at least two vertices, removing any edge $(u,v)$ with $v\in V(G)$ cannot disconnect the graph, and removing an edge from $G$ does not disconnect the graph because all nodes are linked through~$v$.

We claim that $G'$ has an independent 2-dominating set if and only if $G$ has one. 
First, suppose that $G$ has such a set $S$. Note that the set $S$ has at least two selected vertices. Indeed, $G$  has at least one edge, which implies that at least one vertex is not selected (by independence), and such a vertex should be dominated by at least two selected vertices.
Now we claim that $S$ is also a solution for $G'$. Indeed, the addition of $u$ to the graph does not impact the independence of $S$, nor the 2-domination of the nodes of $G$, and $v$ is covered at least twice, since there are at least two selected vertices in $G$.
Second, if $G'$ has independent 2-dominating set $S'$, it cannot contain $v$. 
Indeed, because of the independence condition, if $v$ is selected, then no other node can be selected, and then the 2-domination condition is not satisfied. 
Then $S'$ is contained in $G$ and is 
an independent 2-dominating set of $G$.
\end{proof}

Now, let us formalise the connection  between robustness and independent 2-domination. 

\begin{lemma}\label{lem:perfect-stable-1-robust}
In a 2-connected graph, the 1-robust maximal independent sets are exactly the independent 2-dominating sets.
\end{lemma}

\begin{proof}
As a consequence of Property~\ref{prop:connectivity}, in a 2-connected graph, a 1-robust MIS is an MIS that is robust against the removal of any edge (that is, we can forget about the preserved connectivity in the robustness definition). 
This means that every node not in the MIS is covered twice, otherwise one could break the maximality by removing the edge between the node covered only once and the node that covers it. In other words, the independent set must be 2-dominating.
For the other direction it suffices to note that any independent dominating set is a maximal independent set.   
\end{proof}

At that point, plugging Lemma~\ref{lem:NP-perfect-stable} and Lemma~\ref{lem:perfect-stable-1-robust} we get that deciding whether there exists a 1-robust MIS in a graph is NP-hard, even if we assume 2-connectivity. This last lemma is the final step to prove Theorem~\ref{thm:MIS-NP-hard}. 

\begin{lemma}
For any 2-connected graph $G$ 
and any integer $k>1$, we can build in polynomial-time a graph $G'$, such that: $G$ has a 1-robust MIS if and only if $G'$ has a $2k-1$-robust MIS. 
\end{lemma}

\begin{proof}
We build $G'$ in the following way. Take $k$ copies of $G$, denoted $G_1,..., G_k$, with the notation that $u_x$ is the copy of vertex $u$ in the x-th copy. 
For every edge $(u,v)$ of $G$, we add the edge $(u_x,v_y)$ for any pair $x,y \in 1,...,k$. 

Let us first establish the following claim. An MIS in the graph $G'$ necessarily has the following form: it is the union of the exact same set repeated on each copy. 
Indeed, let $u_i$ be in the MIS. For any $j\ne i$, all the neighbors of $u_j$  in the copy $G_j$ are a neighbor of $u_i$, which implies that they are not in the MIS. 
Hence, no neighbor of $u$ in any copy can be in the MIS. As those nodes are the only neighbors of $u_j$, it implies that $u_j$ is also in the MIS.


Now suppose that $G$ has a 1-robust MIS. We can select the clones of this MIS in each copy, and build an MIS for $G'$ (the independence and maximality are easy to check). 
In this MIS of $G'$, every non-selected vertex has at least $2k$ selected neighbors, therefore this MIS is $2k-1$ robust. 

Finally, suppose that $G'$ has a $2k-1$ robust MIS. Thanks to the claim above, we know that this MIS is the same set of vertices repeated on each copy. We claim that when restricted to a given copy, this MIS is 1-robust. Indeed, if it were not, then there would be one non-selected vertex with at most one selected neighbor, and this would mean that in $G'$ this vertex would have only $k$ selected neighbors, which contradicts the $2k-1$ robust (given the connectivity).

\end{proof}

\section{Conclusions}
\label{sec:concl}

In this paper we have developed the theory of robustness in several ways: adding granularity and studying new natural problems to explore its diversity. 
The next step is to fill in the gaps in our set of results: characterizing exactly the classes $\mathcal{U}^k_{MIS}$, and understanding the complexity of answering the existential question for maximal matching and minimum dominating set. 
We believe that a polynomial-time algorithm can be designed to answer the existential question in the case of maximal matching with $k=1$, with an approach similar to the one of~\cite{CasteigtsDPR20} for MDS (that is, via a careful dynamic programming work on a tree-like decomposition of the graphs).  
A more long-term goal is to reuse the insights gathered by studying robustness to help the design of dynamic algorithms.

\newpage
\paragraph*{Acknowledgements.}

We thank Nicolas El Maalouly for fruitful discussions about perfect matchings and Dennis Olivetti for pointing out reference \cite{Summer79}. 

\DeclareUrlCommand{\Doi}{\urlstyle{same}}
\renewcommand{\doi}[1]{\href{https://doi.org/#1}{\footnotesize\sf doi:\Doi{#1}}}

\bibliographystyle{plainurl}
\bibliography{biblio-robustness}
\end{document}